\documentclass[reqno,12 pt]{amsart}

\usepackage{amsfonts,amscd}
\usepackage{amsthm}
\usepackage{mathrsfs}
\usepackage{amssymb}
\usepackage{eufrak}
\usepackage{latexsym}
\usepackage{enumerate}
\usepackage{cite}
\usepackage{url}
\usepackage{amsmath}
\usepackage{verbatim}
\usepackage{graphicx}
\usepackage{color}
\usepackage{epsf}
\usepackage{geometry}
\usepackage[width=6cm]{caption}
\usepackage{hyperref}
\definecolor{dark-blue}{rgb}{0,0,0.6}
\definecolor{Purple}{rgb}{0.2,0,0.25}
\hypersetup{breaklinks=true,
pagecolor=white,
colorlinks=true,
citecolor=dark-blue,
filecolor=black,%
linkcolor=Purple,%
urlcolor=blue
}
\newcommand{\bref}[1]{\textbf{\ref{#1}}} 
\newcommand{\beqref}[1]{\textbf{(\ref{#1})}} 

\theoremstyle{plain}
\newtheorem{thm}{Theorem}[section]
\newtheorem{lem}[thm]{Lemma}
\newtheorem{defin}[thm]{Definition}
\newtheorem{cor}[thm]{Corollary}
\newtheorem{prop}[thm]{Proposition}
\newtheorem{alg}[thm]{Algorithm}

\theoremstyle{definition}
\newtheorem{remark}[thm]{Remark}

\newtheorem{expl}[thm]{Example}
\newtheorem{unclear-issue}[thm]{Unclear issue}
\newtheorem{unclear-issues}[thm]{Unclear issues}

\newcommand{\dom}{\textnormal{dom}}
\newcommand{\Dom}{\textnormal{Dom}}
\newcommand{\Int}{\textnormal{Int}}

\newcommand{\bisect}{\textnormal{bisect}}

\newcommand{\R}{\mathbb{R}}

\newcommand{\N}{\mathbb{N}}

\title[Computation of zone and double zone diagrams]{On the computation of zone and double zone diagrams}
\author{Daniel Reem}
\address{Department of Mathematics, The Technion - Israel Institute of Technology, 3200003 Haifa, Israel}
\email{dream@technion.ac.il}


\begin{document}
\date{December 31, 2017}

\maketitle
\vspace{0.3cm}
\begin{center}
{ \scriptsize\bf Dedicated to the memory of Ji{\v{r}}{\'{\i}} (Jirka) Matou{\v{s}}ek (1963--2015), an outstanding  scientist, one of the founders of implicit computational geometry}
\end{center}
\vspace{0.5cm}

\begin{abstract}
Classical objects in computational geometry are defined by explicit
relations. Several years ago the pioneering works of T. Asano,
J. Matou\v{s}ek and T. Tokuyama introduced ``implicit computational geometry''
in which the geometric objects are defined by implicit relations involving sets. An
important member in this family is called ``a zone diagram''. The implicit nature of zone  diagrams implies, as has already been observed in the original works, that their computation is a
challenging task. In a continuous setting this task has been addressed (briefly)
only by these authors in the Euclidean plane with point sites.
We discuss the possibility to compute zone diagrams in a wide
class of spaces and also shed new light on their computation in the original
setting. The class of spaces, which is introduced here, includes,
in particular, Euclidean spheres and finite-dimensional strictly convex
normed spaces. Sites of a general form are allowed and it is shown that
a generalization of the iterative method suggested by Asano,
 Matou\v{s}ek and Tokuyama converges to a double zone diagram, another implicit geometric
object whose existence is known in general. Occasionally a zone diagram
can be obtained from this procedure. The actual (approximate) computation of the iterations
is based on a simple algorithm which enables the approximate computation of Voronoi diagrams
in a general setting. Our analysis also yields a few byproducts of independent interest,
such as certain topological properties of Voronoi cells (e.g., that in the considered setting
their boundaries cannot be ``fat'').
\end{abstract}

\section{\bf Introduction}\label{sec:Intro}
\subsection{Background}\label{subsec:Background}
Classical objects in computational geometry, such as polytopes,
arrangements, and Delaunay triangulations are defined by explicit relations
\cite{BoissonnatYvinec1998book,Edelsbrunner-book-1987,GoodmanORourke2004book,Mulmuley1994,SharirAgarwal1995book}. Several years ago the pioneering works \cite{AMT2,AMTn} of T. Asano,
J. Matou\v{s}ek and T. Tokuyama introduced ``implicit computational geometry'', a branch of computational geometry in which the geometric objects are defined by implicit relations involving sets.
An important member in this family is called ``a zone diagram'' (this notion is significantly different from other notions of ``zones'' which exist in computational geometry \cite[Chapter 5]{Edelsbrunner-book-1987}, \cite[p.\,146]{Matousek2002book}, \cite[pp.\,231--236]{SharirAgarwal1995book}).

In order to understand this geometric object better, consider first the much
more familiar concept of a Voronoi diagram, another classical object which is defined explicitly. Suppose first that our setting is the Euclidean plane. We are given a collection of finitely many distinct points $p_1,\ldots,p_n$, $n\in \N$ (called the sites or the  generators) and we associate with each site $p_k$ ($k\in \{1,\ldots,n\}$) its Voronoi cell (or Voronoi region), namely the set of all points $x=(x_1,x_2)$ in $\R^2$ having the property that $d(x,p_k)$, that is, the Euclidean  distance between $x$ and $p_k$, is not greater than  $d(x,p_j)$ for all $j\neq k$, $j\in \{1,\ldots,n\}$. The Voronoi diagram corresponding to the sites $p_1,\ldots,p_n$ is the tuple $(R_1,\ldots,R_n)$ of the corresponding Voronoi cells.  More generally, given a set $X$, a distance function $d\colon X^2\to[0,\infty)$, and a collection of nonempty subsets $(P_k)_{k\in K}$ in $X$,  we associate with each site $P_k$ ($k\in K$) its Voronoi cell, that is, the set $R_k$ of all $x\in X$ whose distance to $P_k$ is not greater than their distance to any of the other sites $P_j$, $j\neq k$, $j\in K$ (exact details appear in Definition \bref{def:Voronoi} below). The collection $(R_k)_{k\in K}$ of Voronoi cells is the Voronoi diagram.

On the other hand, in the case of a zone diagram, we associate with each site $P_k$ the set $R_k$ of all $x\in X$ whose distance to $P_k$ is not greater than their distance to any of the other sets (or cells, or regions) $R_j$, $j\neq k$. The zone diagram induced by these sites is $R=(R_k)_{k\in K}$ (a formal definition appear in  Definition \bref{def:zone} below). Figures \bref{fig:IntroVoronoi} and  \bref{fig:IntroZone} show the Voronoi and zone diagrams, respectively, corresponding to the same ten point sites in the Euclidean plane.
\begin{figure}[t]
\begin{minipage}[t]{0.45\textwidth}
\begin{center}
{\includegraphics[scale=0.54]{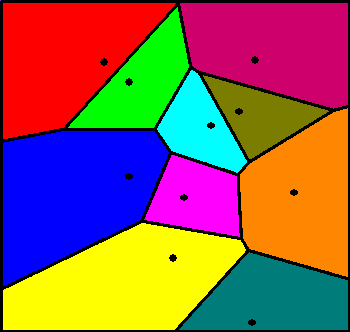}}
\end{center}
 \caption{The Voronoi diagram of 10 point sites in a square in the Euclidean plane.}
\label{fig:IntroVoronoi}
\end{minipage}
\hfill
\begin{minipage}[t]{0.48\textwidth}
\begin{center}
{\includegraphics[scale=0.54]{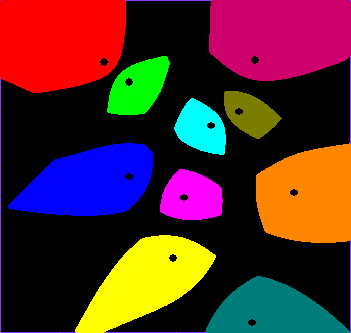}}
\end{center}
 \caption{The zone diagram of the sites of Figure \bref{fig:IntroVoronoi}.}
\label{fig:IntroZone}
\end{minipage}
\end{figure}

At  first glance it seems that the definition of a zone diagram is
circular, because the definition of each region $R_k$ depends on $R_k$ itself via
the definition of the other regions $R_j$, $j\neq k$. On a second thought, we
see that a zone diagram is defined formally to be a fixed point of a
certain mapping (called the $\Dom$ mapping: see Definition \bref{def:zone} below), that is, a solution to the equation $R=\Dom(R)$. Hence, in contrast to the case of Voronoi diagrams which are defined explicitly and hence their existence and uniqueness is obvious, neither the  existence nor the uniqueness of a zone diagram are obvious in advance.
In addition, even if some existence (or uniqueness) results are proved, one still faces the problem of finding algorithms for computing zone diagrams.

Th situation mentioned above is somewhat similar to the situation occurring in the case of  differential equations, where issues related to existence, uniqueness and computation of the solution are frequent \cite{Brezis2011book,CourantHilbert1962IIbook,PinchoverRubinstein2005book,Walter1998book}.  There are, however, a few significant differences between the differential equations setting and our setting, even when a solution of the differential equation (if it exists) induces a geometric object such as a curve or a surface. The first difference is that  the solution in the latter case is a function, which may be defined only locally, and should satisfy some nice properties (e.g., to be smooth or, in the case of weak solutions, to be, say, square integrable). Second, the scientific domain of differential equations is hundreds years old and it is rich with various techniques for solving and analyzing equations.  In contrast, in the case of zone diagrams in particular and implicit computational geometry in general, the equation involves tuple of sets, the solution may be very exotic, and the whole domain of research is in its initial stages in terms of age, techniques, etc.

In their pioneering works, Asano, Matou{\v{s}}ek and Tokuyama introduced and studied zone diagrams. The setting was the Euclidean plane, each site $P_k$ was a single point, and all of these (finitely many) points were distinct. They proved that in this setting there exists a unique zone diagram and they suggested an iterative algorithm for approximating this object. In later works some of these results have been generalized. For instance, in \cite[Theorem 5.6]{ReemReichZone} it was shown that a zone diagram of two general sites
 in any metric space always exists (i.e., the sites $P_1$ and $P_2$ can be any nonempty subsets of the space). In fact, the proof holds in a more  general setting
 called $m$-spaces, in which $X$ is an arbitrary nonempty set and the ``distance'' function $d$
 should only satisfy the condition $d(x,x)\leq d(x,y)\,\,\,\,
\forall x,y\in X$ and can take any value in the interval
$[-\infty,\infty]$. Simple examples given there show that in general uniqueness of a zone
diagram does not necessarily hold.  In \cite{KopeckaReemReich} it is shown, in particular,
that a zone diagram of any finite  number  of sites exists, assuming these sites are compact subsets of the space and they are positively separated (that is, there is a positive number $r$ such that the distance between any two sites $P_k$ and $P_j$, $k\neq j$ is at least $r$) and that they are located in the interior of a large compact subset of a finite dimensional normed space where the norm is strictly convex (the exact  result is more general). Another existence (and also uniqueness) result is discussed in \cite{KMT2012}. Here the setting is a finite dimensional Euclidean space, or, more generally, a finite dimensional normed space which is both strictly convex and smooth  \cite{KMT2012}. The sites are again positively separated.

Zone diagrams have a nice interpretation (introduced in \cite{AMTn} and extended in \cite{ReemReichZone,Reem2014jour}) as a certain equilibrium between mutually hostile kingdoms.
This interpretation is closely related to the fact that a zone diagram induces a neutral region, or a ``no-man's land'', as is illustrated in Figure \bref{fig:IntroZone}. See \cite{Reem2014jour} for a rigorous discussion on this issue. In the case of a discrete  setting there is a different  interpretation: a zone diagram is a certain equilibrium in a  certain combinatorial game involving one player. See \cite{ReemReichZone}.

In addition to zone diagrams, other implicit geometrical objects were introduced and studied,
partly in order to understand better zone diagrams. One such an object is called a double zone diagram  \cite{ReemReichZone}. Formally, it is  defined to be the fixed point of the second iteration of $\Dom$, the mapping which defines zone diagrams. In the language of dynamical systems, double zones diagram can be thought of as being cycles (periodic points) of order two.  It can be shown that any zone diagram must be a double zone diagram, and, as a result, if we are able to show the existence of a double zone diagram, then we have a candidate to be a zone diagram (see also Remark \bref{rem:WeakSolution} below). As a matter of fact, it was proved in \cite[Theorem 5.5]{ReemReichZone} that double zone diagrams exist in a relatively general setting (the setting is an arbitrary $m$-space with arbitrary sites, possibly infinitely many).   Actually, it was shown that there exist a least and a greatest double zone diagrams, namely double zone diagrams $m$ and $M$ such that any other double zone diagram $R$ satisfies $m\subseteq R\subseteq M$ (the inclusion is component-wise). Unfortunately, the proof is based on a nonconstructive argument (the Knaster-Tarski fixed point theorem \cite{Knaster,Tarski}; as a side remark we note that interestingly, the Knaster-Tarski theorem has also applications in logic and computer science \cite{EbbinghausFlum1999book},\cite[Chapters 2--3]{GKLMSVVW2007book}). Hence no general  procedure was suggested to compute the double zone diagrams. The importance of double zone diagrams to the computation of zone diagrams will  become clear later (Sections \bref{sec:QualitativeDescription}--\bref{sec:ConvergenceResult}  below).

One of the main challenges regarding zone diagrams is their computation. This is evident already in the title of the original work \cite{AMTn}.
In a continuous setting this task has been addressed  so far only by  Asano, Matou{\v{s}}ek, and Tokuyama in \cite{AMTn} in the case of the Euclidean plane with finitely many point sites. While the formal claims and proofs mentioned there are insightful,
the discussion about an actual method for the approximate computation of zone diagrams
is very brief with almost no theoretical or  practical details. More precisely, it is written there that  one can use convex polygons with many sides for approximating the components of the iterative sequences, and that each iteration is computationally  demanding, but with the exception of a few interesting
pictures, no additional details can be found there. One of the main difficulties in trying to extend  the method of approximate computation mentioned in \cite{AMTn} to more general settings is the  need for a method for computing Voronoi diagrams in a general setting. Another difficulty is  the lack of any known representation of the boundaries of the involved regions. In fact, they are conjectured to be non-algebraic in many cases and this conjecture is supported by the recent paper  \cite{MonterdeOngay2014jour} which establishes this claim for the case of two point sites in the Euclidean plane.

In a discrete setting ($X$ is a finite set of points) there has been a limited discussion on the issue of computing zone diagrams in two places. First, in \cite{COT} the setting is  digital Euclidean  plane with two sites (a point and a ``line segment'' or a ``curve''; a possible generalization to finitely many line segments was mentioned but no details were given). The actual computation in \cite{COT} is heavily based on the considered setting and  the corresponding analysis (time complexity, etc.) is quite brief. The second place which discusses the computation of zone diagrams in a discrete setting is \cite{ReemReichZone}. There the setting is  a finite $m$-space with two finite sites. The actual computation is by brute force. In both cases mentioned above it is  not clear in which sense the resulting discrete zone diagrams approximate the continuous ones, although  intuitively they may approximate them with respect to the Hausdorff distance, at least in some familiar settings. The computation of double zone diagrams was discussed briefly in \cite{ReemReichZone} only in a discrete setting (finitely many finite sites in an $m$-space which is composed of finitely many points).

In addition to zone and double zone diagrams, other implicit geometric objects have been  studied, among them trisectors \cite{AsanKirk,AMT2,AsanoTokuyama, ChunOkadeTokuyama2010}, $k$-sectors and $k$-gradations \cite{FHKLOMN2013conf,ImaiKawamuraMatousekReemTokuyamaCGTA}, territory diagrams \cite{DeBiasiKalantaris2011} (also called ``mollified zone diagrams'' \cite{DeBiasiKalantaris2010inproc}), and double territory diagrams \cite{Reem2014jour}.  The computation of some of these objects have  been considered in several works, e.g., in \cite{AsanKirk,AMT2,AsanoTokuyama, ChunOkadeTokuyama2010,ImaiKawamuraMatousekReemTokuyamaCGTA} (trisectors and $k$-sectors) and in   \cite{DeBiasiKalantaris2011}  (territory diagrams). A typical feature of the computation is that it is an approximate computation (based on, say, polygonal approximation). One can find interesting pictures and ideas in these works, but in many cases a corresponding precise analysis of convergence/level of approximation/time complexity is very brief or absent.

\subsection{Contributions:}
 This paper considers the question of computing zone diagrams in a new and wide class of metric spaces called ``proper geodesic metric spaces which have the geodesic inclusion property'' (Section  \bref{sec:Geodesic} below). This  class of spaces includes, in particular, Euclidean spheres of any dimension and finite  dimensional strictly convex spaces. The considered sites can be arbitrary  positively separated closed sets (possibly infinitely many). Despite the general setting considered here, our results also shed new light on the computation of zone diagrams in familiar settings such as the Euclidean plane with  point sites. We prove that a generalization of the iterative algorithm suggested by Asano, Matou{\v{s}}ek, and Tokuyama  converges to the least and the greatest double zone diagrams (Theorem \bref{thm:AlgUniConvex} and Corollary \bref{cor:CompactLimit} below). In various cases  a zone diagram can be obtained from the resulting double zone diagrams, and, as a matter of fact, in some of these cases (e.g., when the space is a finite dimensional Euclidean space) the limit is the unique zone diagram (Corollary \bref{cor:ZoneDoubleZone} below). Additional properties are established too (Theorem \bref{thm:Metric} below).

In the normed space case the suggested way to compute approximately the corresponding iterations and the induced (double) zone diagrams is done using the algorithm for computing Voronoi diagrams which was suggested in \cite{ReemISVD09}. This algorithm enables the approximate computation of Voronoi diagrams in a general setting (any  norm, and dimension, sites of a general form). Since this  latter algorithm enables the computation of each Voronoi cell independently of the other ones, the above-mentioned iterative algorithm for computing zone diagrams and double zone diagrams supports, in a  natural way, their computation in a parallel computing environment. Many pictures of (double) zone diagrams and of the corresponding   approximating iterations are given, pictures which have been produced using a computer based  implementation of the algorithm \cite{Vdream2017web} (Section \bref{sec:ActualCompute} below). Our analysis yields a few byproducts of independent interest, among them certain topological properties  of Voronoi cells (that in the setting that we consider, their boundaries cannot be ``fat'': Theorem \bref{thm:BoundaryInterior} below).

\subsection{Paper layout}
In Section \bref{sec:definitions} the basic definitions and notation are presented. In Section \bref{sec:Geodesic} geodesic metric spaces and particular classes of them are discussed. In Section  \bref{sec:QualitativeDescription} we discuss qualitatively issues related to the iterative algorithm. In Section \bref{sec:ConvergenceResult}  the main convergence results are  presented. In Section \bref{sec:ActualCompute} a few  theoretical and practical issues related to the algorithm for computing zone and double zone diagrams are discussed, among them the issue of time complexity and some implementation details regarding the computation of the corresponding  iterations.  The topological properties of Voronoi cells which follow as a byproduct of our approach are described in Section  \bref{sec:TopologicalProperties}.  The proofs of the convergence theorems and other claims are  given in Section \bref{sec:proofs}. We conclude the paper in Section \bref{sec:end} with a discussion on lines for further investigation, interesting open questions, and unexplained phenomena.

\section{Notations and a few definitions}\label{sec:definitions}

Throughout the text we will make use of tuples, the components of which are sets (which are subsets of a given nonempty set $X$). Every operation or relation between such tuples, or on a single tuple, is done component-wise. Hence, for example, if $K\neq \emptyset$ is a set of indices, and if $R=(R_k)_{k\in K}$ and $S=(S_k)_{k\in K}$ are two tuples  of sets, then $R\bigcap S=(R_k\cap S_k)_{k\in K}$,  $\overline{R}=(\overline{R_k})_{k\in K}$, and $R\subseteq S$ means $R_k\subseteq S_k$ for each $k\in K$. When $R$ is the tuple
$(R_k)_{k\in K}$, the  notation $(R)_k$ means the $k$-th component of $R$, i.e, $(R)_k=R_k$. If $S$ is a finite set, then we denote by $|S|$ the number of elements in $S$.

\begin{defin}\label{def:dom}
Let $(X,d)$ be a metric space. Given two nonempty subsets $P,A\subseteq X$, the dominance region
$\dom(P,A)$ of $P$ with respect to $A$ is the set of all $x\in X$
closer (not necessarily strictly) to $P$ than to $A$, i.e.,
$\dom(P,A):=\{x\in X: d(x,P)\leq d(x,A)\}$, where $d(x,A):=\inf\{d(x,a): a\in A\}$. The bisector (or equidistant set) between $P$ and $A$ is the (possibly empty) set $\bisect(P,A):=\{x\in X: d(x,P)=d(x,A)\}$.
\end{defin}
\begin{defin}\label{def:Voronoi}
Let $(X,d)$ be a metric space. Let $K$ be a set of at least two elements (indices), possibly
infinite. Given a  tuple $(P_k)_{k\in K}$ of nonempty subsets
$P_k\subseteq X$, called the generators or the sites, the Voronoi diagram  induced by this tuple is the tuple $(R_k)_{k\in K}$ of nonempty subsets
$R_k\subseteq X$ which satisfy, for all $k\in K$,
\begin{equation*}
R_k=\dom\Bigl(P_k,{\bigcup_{j\neq k} P_j}\Bigr)=\{x\in X: d(x,P_k)\leq d(x,P_j)\,\,\forall j\neq k,\,j\in K \}.
\end{equation*}
 In other words, each $R_k$, called a Voronoi cell or a Voronoi region, is the set of
all $x\in X$ whose distance to the site $P_k$ is not greater than
their distance to any other site $P_j$, $j\neq k$.
\end{defin}
\begin{defin}\label{def:zone}
Let $(X,d)$ be a metric space. Let $K$ be a set of at least two elements (indices), possibly
infinite. Given a  tuple $(P_k)_{k\in K}$ of nonempty subsets
$P_k\subseteq X$, a zone diagram with respect to
that tuple is a tuple $R=(R_k)_{k\in K}$ of nonempty subsets
$R_k\subseteq X$ which satisfy
\begin{equation*}
R_k=\dom\Bigl(P_k,\bigcup_{j\neq k} R_j\Bigr)\quad \forall k\in K.
\end{equation*}
In other words, if one defines $X_k:=\{C: P_k\subseteq C\subseteq X\}$, then a zone diagram
is a fixed point of the mapping $\Dom\colon \prod_{k\in K}
X_k\to \prod_{k\in K}
X_k$, defined by
\begin{equation}\label{eq:TZoneDef}
\Dom(R):=\Bigl(\dom\Bigl(P_k,\bigcup_{j\neq k} R_j\Bigr)\Bigr)_{k\in K},
\end{equation}
that is, $R=\Dom(R)$. A tuple $R=(R_k)_{k\in K}$ is called a double zone diagram if it is a fixed point of the second iteration $\Dom\circ \Dom$, i.e., $R=\Dom^2(R)$.
\end{defin}
If $R=(R_k)_{k\in K}$ is a zone diagram, then $R$ is a double zone diagram
(this can be seen by applying $\Dom$ to the equation  $R=\Dom(R)$). In addition, $P_k\subseteq R_k$ for all $k\in K$ as follows the fact that $P\subseteq \dom(P,A)$ and from Definition \bref{def:zone}.

\begin{defin}\label{def:proper}
A metric space $(X,d)$ is called proper, or finitely compact, if closed balls in $X$ are compact, or, equivalently,
 any bounded sequence has a convergent subsequence.
\end{defin}
Typical examples of proper metric spaces are finite dimensional normed spaces, finite dimensional manifolds, and compact metric spaces.

\begin{defin}\label{def:Hausdorff}
Given two nonempty sets $A_1,A_2$ in a metric space $(X,d)$, the Hausdorff distance between them is defined by
\begin{equation*}
D(A_1,A_2):=\max{\bigl\{\sup_{a_1\in A_1}d(a_1,A_2),\sup_{a_2\in A_2}d(a_2,A_1)\bigr\}}.
\end{equation*}
\end{defin}
Of course, the Hausdorff distance is different from the ``usual'' distance
\begin{equation*}
d(A_1,A_2):=\inf{\{d(a_1,a_2): a_1\in A_1,\,a_2\in A_2\}}.
\end{equation*}

\section{Geodesic metric spaces of various types}\label{sec:Geodesic}
This section discusses geodesic metric spaces. In particular,
it introduces the class of spaces which is central to this paper, namely
geodesic metric spaces which have the geodesic inclusion property.
\begin{defin}\label{def:GeodesicMetric}
Let $(X,d)$ be a metric space.
\begin{enumerate}[(a)]
\item \label{def:GeodesicMetric:GeodesicSegment}
Let $x,y\in S\subseteq X$.  The subset $S$ is
called a geodesic segment (or a metric segment,or, briefly, a segment) between $x$ and $y$  if there exists
 an isometric function $\gamma$ (that is, a distance preserving mapping) which maps
 a  real line segment $[r_1,r_2]$ onto $S$ and satisfies  $\gamma(r_1)=x$ and  $\gamma(r_2)=y$.
 We denote $S=[x,y]_{\gamma}$, or simply $S=[x,y]$.
 If between all points $x,y\in X$ there exists a geodesic segment,
 then $(X,d)$ is called  a  geodesic metric space. The sets $[x,y)=[x,y]\backslash\{y\}$,
 $(x,y]=[x,y]\backslash\{x\}$, and $(x,y)=[x,y]\backslash\{x,y\}$ represent the half-open segments and open
 segments respectively.
\item\label{def:GeodesicMetric:GeodesicInclusion} A geodesic metric space $(X,d)$ is said to have
the geodesic inclusion property if
a nontrivial intersection (namely, an intersection which contains at least two different points) between two different geodesics can have no strict bifurcation point of the form ``$\alpha$''. More precisely, the following holds: given $u,v,b,z\in X$, if  $b\in [u,z)_{\gamma_1}$ and $b\in [v,z)_{\gamma_2}$ (here $b$ is the candidate to be the ``$\alpha$''-like bifurcation point, i.e., in the intersection of the arcs of the ``$\alpha$'' sign), then either $u\in [v,z]_{\gamma_2}$ or $v\in [u,z]_{\gamma_1}$.
\end{enumerate}
\end{defin}

An illustration of Definition \bref{def:GeodesicMetric}\,\beqref{def:GeodesicMetric:GeodesicInclusion} is
given in Figure \bref{fig:GeodesicInclusion}. Note that a simple consequence of the definition is
that actually either $u\in [v,b]_{\gamma_2}$ or $v\in [u,b]_{\gamma_1}$, because if for
instance $u\in (b,z]_{\gamma_2}$, then this and $b\in [u,z]_{\gamma_1}$ imply $d(u,z)<d(b,z)\leq d(u,z)$, a contradiction.

\begin{figure}
\begin{minipage}[t]{1\textwidth}
\begin{center}
{\includegraphics[trim=40 690 0 0, clip=true, scale=0.9]{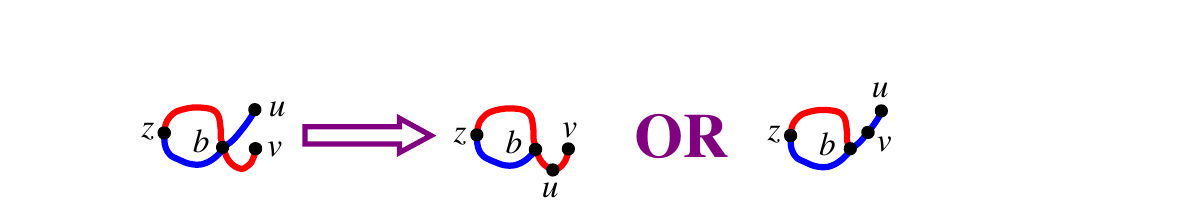}}
\end{center}
 \caption{An illustration of the geodesic inclusion property (Definition \bref{def:GeodesicMetric}\,\beqref{def:GeodesicMetric:GeodesicInclusion}).}
\label{fig:GeodesicInclusion}
\end{minipage}
\end{figure}

\begin{expl}
 Simple and familiar examples of geodesic metric spaces are: the Euclidean plane,
 any nonempty convex subset of a normed space [the geodesic segments are line segments, with $\gamma(t):=x+t(y-x)/|y-x|$, $t\in [0,|x-y|]$ as a possible isometric mapping between $x\neq y$, but when the norm is not strictly
 convex, then other types of geodesic segments exist, e.g., ``zigzag'' ones:
 for instance, $([0,1]\times \{0\})\cup (\{1\}\times [0,1])$ connecting
 $x=(0,0)$ and  $y=(1,1)$ in $(\R^2,\ell_1)$; see also Example \bref{ex:EuclidSphereStrictlyConvex} and Proposition \bref{prop:GeodesicInclusionStrictlyConvex} below],
 Euclidean spheres (a geodesic segment between two points is the shortest arc on a great circle on which the
points are located), complete  Riemannian manifolds \cite[Theorem 1.7.1, p.\,35]{Jost2011book} (in particular, compact ones), hyperbolic spaces \cite[pp.\,538--539]{ReichShafrir},  CAT(0) spaces and Hadamard spaces \cite[pp.\,5--6, Definitions 1.2.1, 1.2.3]{Bacak2014book}.
\end{expl}
\begin{expl}\label{ex:EuclidSphereStrictlyConvex}
The geodesic inclusion property holds for Euclidean
spheres of any dimension (because a nontrivial intersection between two geodesics, namely an at-least-two-points-intersection between short arcs of great circles, can happen only at either two opposite points (opposite poles) or when the geodesics
are identical; but then the points $u$ and $v$ from Definition \bref{def:GeodesicMetric}\,\beqref{def:GeodesicMetric:GeodesicInclusion} coincide).
As shown in Proposition \bref{prop:GeodesicInclusionStrictlyConvex} below, this property also holds for (finite- and infinite-dimensional) strictly convex normed spaces, as well as for non-trivial convex subsets of these spaces; in fact, this proposition shows that a normed space $(X,|\cdot|)$ is strictly convex if and only if it has the geodesic inclusion property. Strictly convex normed spaces are very common in applications. They are characterized by the property that $|x+y|<|x|+|y|$ holds
whenever $x$ and $y$ are arbitrary elements in the space which are  not on the
 same ray (i.e., $x\neq 0, y\neq 0$, and $x/|x|\neq y/|y|$).
 Equivalently,  the unit sphere of the space does not contain any line segment.
 In particular, given $t\in\N$, if we endow $\R^t$ with the $\ell_p$ norm, $p\in (1,\infty)$, then we obtain a space which is both
 strictly convex and proper; $\R^t$ with the $\ell_1$ or $\ell_{\infty}$ norms is
 a typical example of a normed space which is not strictly convex (see  \cite{LindenTzafriri, Prus2001} for more details).
\end{expl}

\section{A qualitative description of the algorithm}\label{sec:QualitativeDescription}
As mentioned before, a tuple $R=(R_k)_{k\in K}$ of regions (subsets) is a zone diagram if it satisfies the fixed point equation $R=\Dom(R)$. A common and natural approach in fixed point theory for the computation of a fixed point of a given mapping $f$ is to use iterations  \cite{GranasDugondji,HB}. More precisely, one starts with some point $y_0$ in the space $Y$ on which $f$  is defined, and starts iterating $f$. A sequence $y_1=f(y_0), \ldots, y_{i+1}=f(y_i),\ldots$ is generated, and one hopes that it converges in some sense to a fixed point of $f$. In general convergence is not guaranteed (just take a point
on the unit circle and apply iteratively a rotation operator on it), but under some assumptions on the mapping and the space it is possible to prove appropriate convergence results.

Returning to our setting, the given mapping is the $\Dom$ mapping. The given  space on which it is defined is $Y:=\prod_{k\in K}
X_k$, where $X_k:=\{C: P_k\subseteq C\subseteq X\}$ for all $k\in K$. In other words, $Y$ is the collection of all tuples (vectors) whose $k$-th component is an arbitrary subset $C$ of the given world $X$ such that $C$ contains the site $P_k$. Hence a natural choice for the starting point $y_0\in Y$ is the collection $(P_k)_{k\in K}$ of the given sites. This discussion leads to the following algorithm.

\begin{alg}\label{alg:InOn}$\,$\\
{\bf\noindent Initialization:} Let $I^{(0)}:=(P_k)_{k\in K}$ and $O^{(0)}:=\Dom (I^{(0)})$.\\
{\bf \noindent Iterative step:} For each $n\in\N$ define $I^{(n)}:=\Dom(O^{(n-1)})$ and $O^{(n)}:=\Dom (I^{(n)})$.
\end{alg}
We refer to $(I^{(n)})_{n=0}^{\infty}$ as the \emph {inner approximation sequence} and to $(O^{(n)})_{n=0}^{\infty}$ as the \emph {outer approximation sequence}. As observed  in  \cite{AMTn}, and a small part also in \cite{Reem2010PhD,ReemReichZone}, it can be shown that $(I^{(n)})_{n=0}^{\infty}$ is an increasing sequence, that
$(O^{(n)})_{n=0}^{\infty}$ is a decreasing sequence, that $I^{(n)}\subseteq O^{(n)}$ for all $n\in \N\cup\{0\}$, and finally, that if $R$ is a zone or a double done diagram, then $I^{(n)}\subseteq R\subseteq O^{(n)}$ for all $n\in \N\cup\{0\}$ (see Lemma
\bref{lem:DomPk} below for a proof). Hence the inner approximation sequence approximates any zone or a double zone diagram from inside (namely, from below), and the outer approximation sequence approximates them from  outside (namely, from above). Figures \bref{fig:Zone2D-10Sites-1inSite-L2-UnitSphere001-It1-EndpointsConnect}--\bref{fig:Zone2D-10Sites-1inSite-L2-UnitSphere001-It4-EndpointsConnect} show $I^{(n)}$ and $O^{(n)}$ (actually, their boundaries) for all $n\in \{1,2,3,4\}$, in a square in $(\R^2,\ell_2)$. The sites are as in Figure \bref{fig:IntroZone} and thus these figures approximate
the zone diagram of Figure \bref{fig:IntroZone}. Other examples illustrating (the boundaries of) $I^{(n)}$ and $O^{(n)}$ can be found in Figures
\bref{fig:Zone2D-10Sites-1inSite-L1-UnitSphere005-It2-Endpoints}--\bref{fig:Zone2D-10Sites-1inSite-L1-UnitSphere005-It3-EndpointsConnect} and \bref{fig:ZoneLinfty}.

\begin{figure}
\begin{minipage}[t]{0.45\textwidth}
\begin{center}
{\includegraphics[scale=0.5]{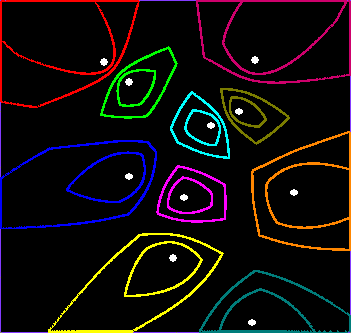}}
\end{center}
 \caption{Approximation of the zone diagram of Figure  \bref{fig:IntroZone} using $I^{(1)}$
 and $O^{(1)}$ (800 endpoints approximate the boundary of each region, and neighbor endpoints are connected by a line segment).}
\label{fig:Zone2D-10Sites-1inSite-L2-UnitSphere001-It1-EndpointsConnect}
\end{minipage}
\begin{minipage}[t]{0.45\textwidth}
\begin{center}
{\includegraphics[scale=0.5]{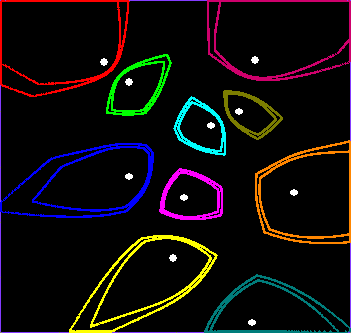}}
\end{center}
 \caption{The setting of Figure \bref{fig:Zone2D-10Sites-1inSite-L2-UnitSphere001-It1-EndpointsConnect}, but with $I^{(2)}$ and $O^{(2)}$. }
\label{fig:Zone2D-10Sites-1inSite-L2-UnitSphere001-It2-EndpointsConnect}
\end{minipage}
\hfill
\begin{minipage}[t]{0.45\textwidth}
\begin{center}
{\includegraphics[scale=0.5]{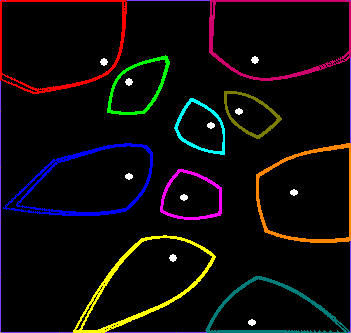}}
\end{center}
 \caption{The setting of Figure \bref{fig:Zone2D-10Sites-1inSite-L2-UnitSphere001-It1-EndpointsConnect}, but with $I^{(3)}$ and $O^{(3)}$. }
\label{fig:Zone2D-10Sites-1inSite-L2-UnitSphere001-It3-EndpointsConnect}
\end{minipage}
\begin{minipage}[t]{0.45\textwidth}
\begin{center}
{\includegraphics[scale=0.5]{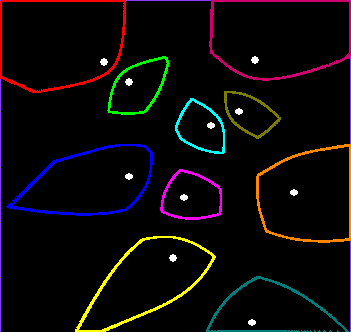}}
\end{center}
 \caption{The setting of Figure \bref{fig:Zone2D-10Sites-1inSite-L2-UnitSphere001-It1-EndpointsConnect}, but with $I^{(4)}$ and $O^{(4)}$. }
\label{fig:Zone2D-10Sites-1inSite-L2-UnitSphere001-It4-EndpointsConnect}
\end{minipage}

\end{figure}

At this stage several difficulties arise. First, it is not clear that these sequences converge, and if they do converge, it is not clear whether both of them converge to the same limit and whether the limit of at least one of them is indeed a zone diagram (as Theorem \bref{thm:AlgUniConvex} below shows, under certain conditions the limits are actually the least and the greatest double zone diagrams).  Second, even if a convergence result is obtained, one faces the problem of the actual computation of $I^{(n)}$ and $O^{(n)}$ for $n\geq 1$. Even for $n=0$ this is not a simple task since $O^{(0)}=\Dom\bigl(I^{(0)}\bigr)=\bigl(\dom\bigl(P_k,\bigcup_{j\neq k}P_j\bigr)\bigr)_{k\in K}$ is the Voronoi diagram of the sites $(P_k)_{k\in K}$, and hence, for sites of a general form, or sites which possibly form a degenerate configuration, or for a space with a general norm, one has to find algorithms which  enable the computation of Voronoi diagrams in such a setting, and most of the familiar algorithms for computing Voronoi diagrams (see e.g., \cite{Aurenhammer,AurenhammerKlein,ComputeGeoBook_BCKO,Edelsbrunner-book-1987,Mulmuley1994,OBSC,SharirAgarwal1995book,YapSharmaLien2012inproc} for some  reviews) are not able to achieve this task (or, in some cases,
they are too complicated or too  slow).

Even if one restricts oneself to the familiar case of the Euclidean plane with point sites, one still faces problems starting from the iteration $n=1$. This is because one has to  know how to compute the components of $\Dom(S)$ for a given tuple $S=(S_k)_{k\in K}$, and hence one has to know a method for computing  $\dom\bigl(P_k,\bigcup_{j\neq k}S_j\bigr)$, i.e., the Voronoi cell of $P_k$ with respect to the set $\bigcup_{j\neq k}S_j$. Unfortunately, when $S=I^{(n)}$ or $S=O^{(n)}$, $n\in\N$, the components of $S$ are general sets, and again, it seems that most of the familiar algorithms for computing Voronoi diagrams  are not helpful here. In Section \bref{sec:ActualCompute} below we explain how to overcome these difficulties.

The sequences $(I^{(n)})_{n=0}^{\infty}$ and $(O^{(n)})_{n=0}^{\infty}$ were introduced in \cite{AMTn} in the case of the Euclidean plane with finitely many point sites. It was shown that $\overline{\bigcup_{n=0}^{\infty} I^{(n)}}=\bigcap_{n=0}^{\infty}O^{(n)}$ and that this tuple is the unique zone diagram.
The proof of this result is not at all obvious and it contains several insightful ideas (see also Section
\bref{sec:ConvergenceResult} below). Although this is not a pure  convergence result (no limits) and although
no error estimates for the level of approximation were given, one
still has the pleasant phenomenon that $(I^{(n)})_{n=0}^{\infty}$ increases to a zone diagram and $(O^{(n)})_{n=0}^{\infty}$ decreases to it.  As for the actual computation of $I^{(n)}$ and $O^{(n)}$, we already mentioned in
Section \bref{sec:Intro} that very few details were given in \cite{AMTn}.

\section{The convergence theorems}\label{sec:ConvergenceResult}
As mentioned in Section \bref{sec:QualitativeDescription}, it was proved in \cite{AMTn} that in the case of the Euclidean plane with point sites one has the equality $m:=\overline{\bigcup_{n=0}^{\infty}I^{(n)}}=\bigcap_{n=0}^{\infty}O^{(n)}:=M$, and $m=M$ is the unique zone diagram. As Theorem \bref{thm:Metric} below shows, it is always true that $M=\Dom(m)$
in any metric space and for arbitrary sites. It is not known that $m=\Dom(M)$ in general, but,
as Theorem \bref{thm:AlgUniConvex} shows,  under the assumptions that the space is a proper geodesic metric space with the geodesic inclusion property and the (possibly infinitely many general) sites are positively separated, the equality $m=\Dom(M)$ holds. However, in this latter case both $m$ and $M$ are double zone diagrams rather than  zone diagrams (actually $m$ is the least double zone diagram and $M$ is the greatest one, that is, $m\subseteq R\subseteq M$ for every double zone diagram $R$) and hence they are not necessarily equal as Example \bref{ex:mM} below shows.
In addition, Theorem \bref{thm:AlgUniConvex} also discusses another way to obtain a zone diagram in the special case of two sites.

A corollary to the theorem ensures that  when the world $X$ is compact, then $(I^{(n)})_{n=0}^{\infty}$ and $(O^{(n)})_{n=0}^{\infty}$ converge to $m$ and $M$ respectively  with respect to the Hausdorff distance. As a result, since in practice $X$ is taken to be compact (e.g., a large box), if $I^{(n)}$ and $O^{(n)}$ are shown experimentally to be almost the same for some $n$, then one has a good approximation to both $m$ and $M$. When it is known that a zone diagram exists, for instance, under the assumptions of \cite{KopeckaReemReich}, then one has a good approximation to this zone diagram, and in fact this also shows that probably this zone diagram is unique and coincides with both $m$ and $M$.  Another corollary to the above theorem ensures that whenever the least and the greatest double zone diagrams coincide, then $m=M$ and  they both coincide with the unique zone diagram. In particular, this is true when the space is Euclidean.

\begin{thm}\label{thm:Metric}
Let $(X,d)$ be a metric space and let $P=(P_k)_{k\in K}$ be a tuple of nonempty sets in $X$. Consider the inner and outer approximation sequences $(I^{(n)})_{n=0}^{\infty}$ and $(O^{(n)})_{n=0}^{\infty}$, respectively (Algorithm \bref{alg:InOn} above) and let
\begin{equation}\label{eq:mM}
M:=\bigcap_{n=0}^{\infty}O^{(n)}, \quad\quad m:=\overline{\bigcup_{n=0}^{\infty}I^{(n)}}.
\end{equation}
Then $M=\Dom(m)$.
\end{thm}

\begin{thm}\label{thm:AlgUniConvex}
Let $(X,d)$ be a proper geodesic metric space which has the geodesic
inclusion property. Suppose that $(P_k)_{k\in K}$ is a given tuple
of closed subsets of $X$ satisfying
\begin{equation}\label{eq:d(P_k,P_j)>0}
 \inf{\{d(P_k,P_j): j,k\in K,\, j\neq k\}}>0.
\end{equation}
Let $m$ and $M$ be defined by \beqref{eq:mM}.
Then $M=\Dom(m)$, $m=\Dom(M)$, and $m$ and $M$ are, respectively, the least  and  greatest double zone diagrams.
In addition, if $|K|=2$, then by letting  $m=(m_1,m_2)$ and $M=(M_1,M_2)$, each of the
pairs $(m_1,M_2)$ and $(M_1,m_2)$ is a zone diagram.
\end{thm}
\begin{cor}\label{cor:CompactLimit}
Under the setting of Theorem \bref{thm:AlgUniConvex}, let
$M=(M_k)_{k\in K}$ and $m=(m_k)_{k\in K}$ be defined by \beqref{eq:mM}. Assume also
that $(X,d)$ is compact. Then, with respect to the Hausdorff distance,
\begin{equation}\label{eq:Mm_k}
M_k=\lim_{n\to\infty}O^{(n)}_k
\quad \textnormal{and} \quad
m_k=\lim_{n\to\infty}I^{(n)}_k\quad \forall k\in K.
\end{equation}
\end{cor}
\begin{cor}\label{cor:ZoneDoubleZone}
Under the setting of Theorem \bref{thm:AlgUniConvex}, if it is known that
 the least and the greatest double zone diagrams coincide, then $m=M$ and they both coincide with the unique
  zone  diagram. In particular this is true when $(X,d)$ is a finite dimensional Euclidean space.
\end{cor}

The proof of the above-mentioned theorems and corollaries is quite long and technical and it can be found in Section \bref{sec:proofs} below.
The proof of  Corollary \bref{cor:CompactLimit} follows from a quite general argument not related to zone diagrams. As already mentioned, the proof given in \cite{AMTn} for the 2D Euclidean case of point sites is
far from being obvious and it contains certain useful ideas which can be modified to the
setting considered here. In particular, this is true for \cite[Lemma 3.1, Lemma 5.1]{AMTn}.
However, the generalization of some of the arguments given there to the setting of this paper is definitely
not immediate and one has to pay attention to certain subtle points, among them the verification of the equality
\begin{equation}\label{eq:UnionIntersectBody}
\bigcup_{j\neq  k}\bigcap_{\gamma=1}^{\infty}\,(\Dom^{2\gamma+1}(P_k)_{k\in K})_j=\bigcap_{\gamma=1}^{\infty}\bigcup_{j\neq k}\, (\Dom^{2\gamma+1}(P_k)_{k\in K})_j
\end{equation}
in the setting discussed in this paper (this equality is implicit in the proof of \cite[Lemma 5.1]{AMTn} but its derivation there is based on the setting discussed there).
One may wonder regarding additional differences between the proof in the setting of \cite{AMTn}
and the one considered here. Essentially, the main tools established here
are the derivation of certain properties of $\dom$ and $\Dom$ (e.g., Lemma \bref{lem:Dom3} below), all of them seem to be new but some of them (see below) generalize known results.
We feel that when known results are generalized, the approach given here either illuminates
certain implicit or explicit key points discussed elsewhere or establishes new tools not
discussed elsewhere. Again, even in the case where a generalization is made,
because of the general setting considered here there are several difficulties and
subtle points that should be handled correctly and one example was mentioned
above regarding \beqref{eq:UnionIntersectBody}. One may also wonder whether it is
possible to simplify the proof in the specific case of the Euclidean plane (or at least
for Euclidean spaces). Unfortunately, it seems that the answer is no, partly because (as mentioned
above) already the Euclidean proof is not simple. But this can be regarded as an advantage
because it shows that the arguments given here use only certain mild but important properties
which certain spaces have, among them Euclidean spaces.

As a final remark regarding this issue, it may be interesting to note that there are
some connections between some auxiliary tools used here and in  \cite{ImaiKawamuraMatousekReemTokuyamaCGTA},
e.g., \cite[Proposition 5]{ImaiKawamuraMatousekReemTokuyamaCGTA}. In addition, the setting of the existence assertion in
\cite{ImaiKawamuraMatousekReemTokuyamaCGTA} (but not of the computational part, namely
\cite[Section 4]{ImaiKawamuraMatousekReemTokuyamaCGTA}, in which the setting is a finite dimensional strictly convex normed space) is proper geodesic metric spaces. Nevertheless
 the proofs of \cite[Proposition 5]{ImaiKawamuraMatousekReemTokuyamaCGTA} and the assertions
 mentioned above (e.g., Theorem \bref{thm:AlgUniConvex}) are different since the involved mappings are not the same  and hence one needs to find a strategy which is appropriate for each case separately. An important feature used in the proof of Theorem \bref{thm:AlgUniConvex} is the fact that the sites are positively separated (this also allows us to consider infinitely many sites) while in  \cite[Proposition 5]{ImaiKawamuraMatousekReemTokuyamaCGTA} the sites satisfy the more general condition of being merely disjoint, but now one must consider only two sites (by definition) and the mapping must have finitely many components. In addition, it seems that some of the auxiliary
assertions described in \cite[Section 4]{ImaiKawamuraMatousekReemTokuyamaCGTA} can
actually be generalized to proper metric spaces having the geodesic inclusion property using  some of the tools mentioned in Section  \bref{sec:proofs} below (see also Section \bref{sec:TopologicalProperties} below).

\begin{figure}
\begin{minipage}[htb]{0.49\textwidth}
\begin{center}
{\includegraphics[scale=0.51]{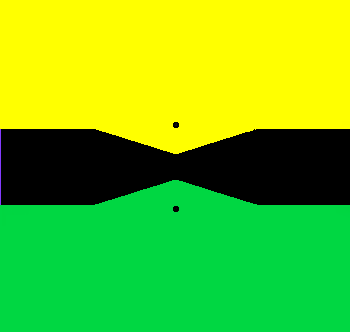}}
\end{center}
 \caption{The components of the greatest double zone diagram $M$ mentioned in Example \bref{ex:mM}.}
\label{fig:GreatestDoubleZoneStrangeNorm}
\end{minipage}
\begin{minipage}[htb]{0.49\textwidth}
\begin{center}
{\includegraphics[scale=0.51]{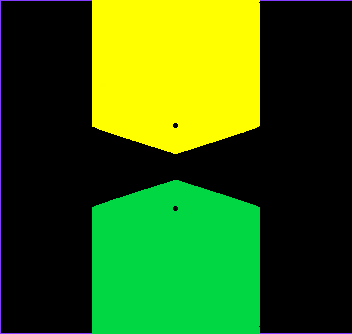}}
\end{center}
 \caption{The components of the least double zone diagram $m$ mentioned in Example \bref{ex:mM}.}
\label{fig:LeastDoubleZoneStrangeNorm}
\end{minipage}
\hfill
\begin{minipage}[t]{0.45\textwidth}
\begin{center}
{\includegraphics[scale=0.51]{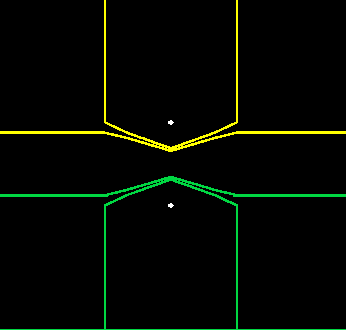}}
\end{center}
 \caption{$I^{(2)}$ and $O^{(2)}$ for the double zone diagrams mentioned in Example \bref{ex:mM} (where neighbor endpoints are connected by solid lines).}
\label{fig:ZD-Iteration2-StrangeNorm-2Sites-0002}
\end{minipage}
\begin{minipage}[t]{0.48\textwidth}
\begin{center}
{\includegraphics[scale=0.5]{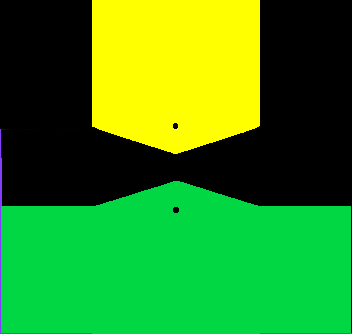}}
\end{center}
 \caption{The zone diagram $(m_1,M_2)$ mentioned in Example \bref{ex:mM}.}
\label{fig:ZD-StrangeNorm-2Sites-0002}
\end{minipage}
\end{figure}

We finish this section by discussing briefly an example which illustrates the phenomenon of non-uniqueness of zone and double zone diagrams.
\begin{expl}\label{ex:mM}
An illustration of Theorem \bref{thm:AlgUniConvex} for the case where $m$ and $M$ from \beqref{eq:mM} satisfy $m\neq M$ is given in Figures \bref{fig:GreatestDoubleZoneStrangeNorm}--\bref{fig:LeastDoubleZoneStrangeNorm}.
Here the sites are $P_1=\{(0,1)\}$ and $P_2=\{(0,-1)\}$. They are located in a rectangle in $\R^2$ endowed with the
following norm:
\begin{equation}\label{eq:alpha delta}
\|(x_1,x_2)\|:=\delta\sqrt{\alpha^2|x_1|^2+|x_2|^2}+(1-\alpha\delta)|x_1|+(1-\delta)|x_2|,\,\, \forall (x_1,x_2)\in\R^2,
\end{equation}
where $\alpha$ and $\delta$ are fixed numbers satisfying $\alpha\in [0,\infty)$, $\delta\in [0,1]$ and $\alpha\delta\in [0,1]$ (it can be checked easily that $\|\cdot\|$ is indeed a norm by separating into cases according to the values of $\alpha$ and $\delta$). In Figures \bref{fig:GreatestDoubleZoneStrangeNorm}--\bref{fig:ZD-StrangeNorm-2Sites-0002} we took $\alpha=\delta=0.1$. The norm \beqref{eq:alpha delta} is a certain mixture between the Euclidean norm (obtained for $\alpha=\delta=1)$ and the $\ell_1$ norm (obtained for $\alpha=0$ or $\delta=0$). When $\alpha,\delta\in (0,1)$, then $\|\cdot\|$ is strictly convex but not smooth (the unit sphere contains points which have several supporting lines).
A closely related example was discussed in \cite[Section 5, Appendix C]{KMT2012} in the context of
non-uniqueness of zone diagrams (double zone diagrams were not mentioned). It follows
from the discussion presented there that indeed $m\neq M$ if $\alpha$ and $\delta$ are sufficiently small. Two different zone diagrams can be obtained in this case: either $(m_1,M_2)$ or $(M_1,m_2)$. The first of them is shown in Figure  \bref{fig:ZD-StrangeNorm-2Sites-0002}.

Figure \bref{fig:ZD-Iteration2-StrangeNorm-2Sites-0002} shows an approximation of $m$ and $M$ using $I^{(2)}$ and $O^{(2)}$. For producing this figure the endpoints of 4000 rays emanating from each site were computed in the Voronoi algorithm stage (see Section \bref{sec:ActualCompute} below for more details about this stage). After this stage neighbor endpoints were connected by a solid line. Without connecting the neighbor endpoints some parts of the boundaries of the components of $I^{(2)}$ and $O^{(2)}$ appear as not being full. This happens  because many rays should be produced in a very small angle (i.e., to belong to the intersection of a very narrow cone with the unit sphere) due to the location of the site with respect to these boundary parts (these parts and the site are almost located on the same line, thus many rays should be produced in directions very close to the two possible directions of the line). Instead of (or complementary to) connecting neighbor endpoints by a line segment, one can produce the rays in a non-uniform way, namely to select unit vectors on the unit sphere of the space in a non-uniform way such that many rays will emanate in the corresponding sector but relatively few in other sectors, in contrast to the current way of producing the rays in a roughly uniform way.
\end{expl}

\section{The actual computation of $I^{(n)}$ and $O^{(n)}$ and related analysis}\label{sec:ActualCompute}

In this section we provide details about practical and theoretical aspects of the computation method. For the sake of convenience, the section is divided into several subsections. A corresponding implementation (including a source code) can be found in \cite{Vdream2017web}. (As a side remark we note that this implementation can do additional tasks of independent interest, such as simulating a crystal growth of the Voronoi cells and, as a byproduct, illustrating the unit ball of the considered normed space.)

\subsection{The computational model}\label{subsec:ComputationalModel} The involved geometrical objects considered here have a complicated and also non-classical nature. For instance, their definition is based on an implicit relation, the inducing sites can be complicated, the space can be of any (possibly high) dimension, the computed regions can have exotic shapes, the boundaries of the computed objects may be non-algebraic even in very simple settings  \cite{MonterdeOngay2014jour} and probably in many more settings, etc. Hence it is hard to
use directly a conventional model such as a linear interpolation model (which is used for approximating  a usually nonlinear curve by a polygonal one) or an algebraic model in which the underlying  structure is algebraic. The approach considered here is based on an approximation model involving shooting rays from the sites in various directions. It can be thought of as being a ``polar model'' or a ``spherical model'' or a ``constant orientation model''. The advantage of this model is that a significant part of the difficulty in the computation is reduced to a one dimensional setting (the ray on which one has to compute the so-called 
``endpoint'': see Subsection \bref{subsec:InOn} below). This reduction enables a simple approximate computation of the regions (of the given iterations $I^{(n)}$ and $O^{(n)}$) and to any required precision, and hence the advancement from the existing theory of zone diagrams is not only the  general setting consider here (e.g., rather general distance functions, general sites) but also a progress in more familiar settings such as the Euclidean plane with point sites (as explained in Subsection \bref{subsec:Background} and in Section \bref{sec:QualitativeDescription} above, it is a challenging task to compute zone diagrams even in this simple setting). Once the above-mentioned model is used, it also enables the use of the familiar polygonal approximation model (i.e., a linear interpolation model) in a natural and simple way: see Subsection \bref{subsec:InOn}  below. In what follows we discuss our model in the case where our world (the geodesic metric space) is a compact subset of a normed space. The reason for doing this is because in this case we have a working implementation \cite{Vdream2017web}, but in principle the details below can be generalized (a related brief and preliminary discussion can be found in \cite[Section 2.6]{Reem2010PhD}).

\subsection{\bf The actual computation of $I^{(n)}$ and $O^{(n)}$:}\label{subsec:InOn}
As mentioned in Section \bref{sec:QualitativeDescription}, in order to compute the iterations $I^{(n)}$ and $O^{(n)}$ one has to know how to compute, or at least to approximate, the dominance regions $\dom(P,A)$ induced by general nonempty sets $P,A$. The way we choose to overcome this difficulty is to use the approximation algorithm for computing Voronoi diagrams of general sites in general normed spaces which was introduced in \cite{ReemISVD09}. In a nutshell, for approximating $\dom(P,A)$ one uses a certain ray-shooting technique, based on the fact that $\dom(P,A)$ can be represented as a union of  rays emanating from the points of $P$
(see e.g., Figures \bref{fig:VoronoiL2718281828}--\bref{fig:ZoneL2718281828}).

\begin{figure}
\begin{minipage}[t]{0.45\textwidth}
\begin{center}
{\includegraphics[scale=0.5]{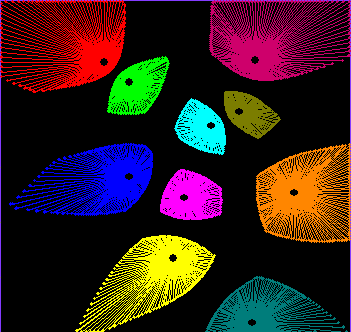}}
\end{center}
 \caption{An approximation of the zone diagram of Figure \bref{fig:IntroZone} (five iterations, 160 rays approximate each region).}
\label{fig:Zone2D-10Sites-1inSite-L2-UnitSphere005-It5-Rays}
\end{minipage}
\begin{minipage}[t]{0.45\textwidth}
\begin{center}
{\includegraphics[scale=0.5]{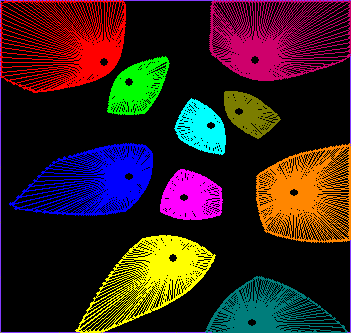}}
\end{center}
 \caption{As in Figure \bref{fig:Zone2D-10Sites-1inSite-L2-UnitSphere005-It5-Rays}, where neighbor endpoints are connected by a line segment.}
\label{fig:Zone2D-10Sites-1inSite-L2-UnitSphere005-It5-RaysEndpoints}
\end{minipage}
\hfill
\begin{minipage}[t]{0.45\textwidth}
\begin{center}
{\includegraphics[scale=0.5]{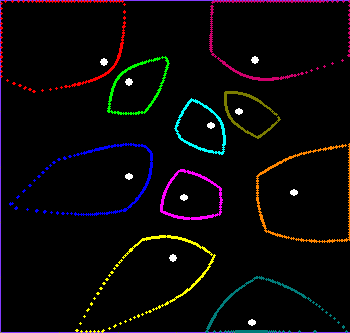}}
\end{center}
 \caption{An approximation of the zone diagram of Figure \bref{fig:IntroZone} (five iterations, 160 endpoints approximate each region).}
\label{fig:Zone2D-10Sites-1inSite-L2-UnitSphere005-It5-Endpoints}
\end{minipage}
\begin{minipage}[t]{0.45\textwidth}
\begin{center}
{\includegraphics[scale=0.5]{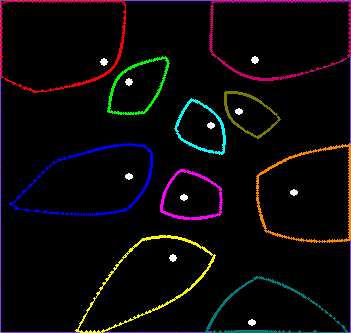}}
\end{center}
 \caption{As in Figure \bref{fig:Zone2D-10Sites-1inSite-L2-UnitSphere005-It5-Endpoints}, where neighbor endpoints are connected by a line segment.}
\label{fig:Zone2D-10Sites-1inSite-L2-UnitSphere005-It5-EndpointsConnect}
\end{minipage}
\end{figure}


First, the world $X$ is assumed to be a large compact subset, e.g., a rectangle or a  hyperbox.  Now one approximates $P$ using a finite collection of points $p$; this is always possible if  $P$ is compact. After that a  finite collection of directions is chosen (that is, a collection of  unit vectors, possibly uniformly distributed on the unit sphere), then one approximates the endpoints of the rays emanating from the points of $P$ in these directions up to any required  precision using the technique described in \cite{ReemISVD09}. At this stage $\dom(P,A)$ is represented by this collection of rays and it is regarded as being computed (of course, this is only an approximation of the real set $\dom(P,A)$). In practice, for each point $p\in P$ one stores the endpoints of the rays emanating from $p$. When producing pictures, one can simply draw the whole ray between each $p\in P$ and each of its endpoints, or one can draw only the endpoints. In dimension 2 one can also draw the endpoints and then connect neighbor endpoints (i.e., endpoints corresponding to neighbor unit vectors) by a line segment so that a polygonal approximation (a mesh) of $\dom(P,A)$ will be obtained  (in higher dimensions one can uses a mesh based on hyperboxes: this is also easily done). This representation is illustrated in various figures, among them Figures \bref{fig:Zone2D-10Sites-1inSite-L2-UnitSphere005-It5-Rays}--\bref{fig:Zone2D-10Sites-1inSite-L2-UnitSphere005-It5-EndpointsConnect}, \bref{fig:VoronoiL2718281828}--\bref{fig:ZoneL2718281828}.  When using this polygonal approximation, one can have a new representation to $\dom(P,A)$ as a collection of points, by replacing the  original endpoints by the centers (or any other intermediate point) of the corresponding polygonal edges which connect two neighbor endpoints.

For computing $I^{(n)}$ and $O^{(n)}$ one  computes their corresponding components iteratively: the components are $\dom(P_k,A_k)$ where $k$ runs over all the indices in $K$ (in practice $K$ is a finite set, that is, $|K|<\infty$) and $A_k$ depends on the iteration and on $k$. In iteration $n+1$ we have $A_k=\bigcup_{j\neq k}(I^{(n)})_k$ or $A_k=\bigcup_{j\neq k}(O^{(n)})_k$ where $(I^{(n)})_k$ and $(O^{(n)})_k$  represent the $k$-th component of $I^{(n)}$ and $O^{(n)}$, respectively, and $n\in\N\cup\{0\}$. For instance, the $k$-th component of $I^{(1)}$ is $\dom(P_k,A_k)$, where $A_k=\bigcup_{j\neq k} (O^{(0)})_k$ and $(O^{(0)})_k=\dom\bigl(P_k,\bigcup_{j\neq k}P_j\bigr)$. Thus, at each iteration we know $A_k$, as a collection of endpoints, from a previous computation. Most of the figures of this paper (in particular, Figures \bref{fig:Zone2D-10Sites-1inSite-L2-UnitSphere001-It1-EndpointsConnect}--\bref{fig:Zone2D-10Sites-1inSite-L2-UnitSphere001-It4-EndpointsConnect}, \bref{fig:Zone2D-10Sites-1inSite-L2-UnitSphere005-It5-Rays}--\bref{fig:Zone2D-10Sites-1inSite-L2-UnitSphere005-It5-EndpointsConnect}, and \bref{fig:Zone2D-10Sites-1inSite-L1-UnitSphere005-It2-Endpoints}--\bref{fig:ZoneLinfty}) were produced using this computational process.
%

\begin{figure}
\begin{minipage}[t]{0.45\textwidth}
\begin{center}
{\includegraphics[scale=0.6]{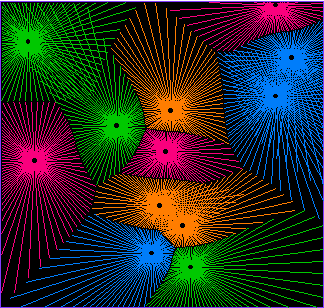}}
\end{center}
 \caption{An approximation of the Voronoi diagram of four sites in $(\R^2,\ell_p)$, where $p\approx 2.71$;
 each site consists of 3 points and 84 rays  emanate from each point of each site.}
\label{fig:VoronoiL2718281828}
\end{minipage}
\hfill
\begin{minipage}[t]{0.45\textwidth}
\begin{center}
{\includegraphics[scale=0.6]{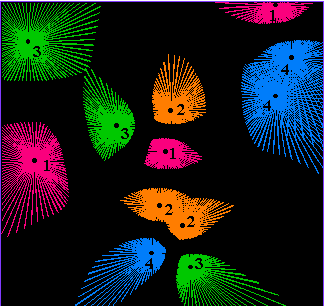}}
\end{center}
 \caption{An approximation of the zone diagram (and hence of the double zone diagram) of the sites of Figure \bref{fig:VoronoiL2718281828}; again 84 rays  emanate from each point of each site.}
\label{fig:ZoneL2718281828}
\end{minipage}
\end{figure}

An important parameter which determines the level of approximation of the regions is the number of rays (or, actually, the density of the corresponding unit vectors in the unit sphere of the space) used in the construction of the dominance regions. In Figures \bref{fig:Zone2D-10Sites-1inSite-L2-UnitSphere001-It1-EndpointsConnect}--\bref{fig:Zone2D-10Sites-1inSite-L2-UnitSphere001-It4-EndpointsConnect} the number of rays used for approximating each region is 800, and only the endpoints are displayed, where neighbor endpoints are connected by a line segment. The number of rays used for approximating each region in Figure \bref{fig:ZD-Iteration2-StrangeNorm-2Sites-0002} is 4000. Again, only endpoints are displayed and neighbor endpoints are connected by a line segment. In Figures \bref{fig:Zone2D-10Sites-1inSite-L2-UnitSphere005-It5-Rays}--\bref{fig:Zone2D-10Sites-1inSite-L2-UnitSphere005-It5-EndpointsConnect} we used 160 rays for approximating each region, and sometimes either the whole ray or only the endpoints are displayed (neighbor endpoints may or may not be connected by a line segment). From each point of each site in Figures  \bref{fig:VoronoiL2718281828}--\bref{fig:ZoneL2718281828} emanate 84 rays, and the whole ray is displayed. Despite the difference in the level of approximation,  the various approximations of the  zone diagram of Figure \bref{fig:IntroZone} (e.g., Figures \bref{fig:Zone2D-10Sites-1inSite-L2-UnitSphere001-It4-EndpointsConnect}, \bref{fig:Zone2D-10Sites-1inSite-L2-UnitSphere005-It5-Rays}--\bref{fig:Zone2D-10Sites-1inSite-L2-UnitSphere005-It5-EndpointsConnect}) yield similar shapes.

\subsection{\bf Time complexity:}\label{sec:complexity}
In this subsection we present the time complexity of the approximation method described in previous subsections. The worst case time complexity is determined by the time complexity associated with the computation of the dominance regions $\dom(P_k,A_k)$, $k\in K$.
Essentially, it is determined by the number of distance comparisons required to compute the
entire number of endpoints. For every $k\in K$ the number of rays emanating from each point of $P_k$ is $C_1C_2^{\rm dim-1}$, where
$\rm dim$ is the dimension, $C_1=C_1(\rm dim)>0$ is a constant depending on the norm (and, through the norm, also on the dimension), and $C_2>0$ is a constant depending on the level of  approximation of the unit sphere  (roughly speaking, $1/C_2$ is a bound on the largest  possible distance between a point on the unit sphere and the set of unit vectors created by the user). Hence the total number of rays emanating from the points of $P_k$ is $C_1C_2^{\rm dim-1}|P_k|$.
Starting from iteration $n=1$, each region is represented by the collection of its endpoints.
Since there are $|K|-1$ regions in $A_k$, namely either $(I^{(n-1)})_j$, $j\in K\backslash \{k\}$ or $(O^{(n-1)})_j$, $j\in K\backslash \{k\}$,  the number of endpoints which represent $A_k$ is
$C_1 C_2^{\rm dim-1}\sum_{j\neq k}|P_j|$. For determining an endpoint up to some user-defined error parameter, one makes $C_3 |A_k|$ distance  comparisons along a given ray (emanated from some point in $P_k$), where $C_3>0$ is a constant depending on the level of approximation.

We conclude from the above-mentioned discussion that the number of distance comparisons done for computing $\dom(P_k,A_k)$ is $C_3|A_k|C_1C_2^{{\rm dim}-1}|P_k|$. Thus the  time complexity for computing all the regions in a given iteration (either an inner approximation or an outer approximation) is $O(\sum_{k=1}^{|K|}C_3|A_k|C_1C_2^{{\rm dim}-1}|P_k|)$, that is, the complexity is $O\bigl(C_3(C_1C_2^{\rm dim-1})^2\sum_{k=1}^{|K|}\bigl(|P_k|\sum_{j\neq k}|P_j|\bigr)\bigr)$ (here, of course, $O(\cdot)$ is the standard big-Oho notation for upper bounds and it is not related to the outer approximation iteration $O^{(n)}$). Now we observe that the same number of computations is done in each iteration and that each iteration consists of two sub-iterations (one for computing the inner approximation and one for the outer approximation). Consequently, the time complexity for computing $I^{(1)}$,$O^{(1)},\ldots, I^{(n)}$, $O^{(n)}$ is $O\bigl(2n\cdot C_3(C_1C_2^{\rm dim-1})^2\sum_{k=1}^{|K|}|P_k|\sum_{j\neq k}|P_j|\bigr)$.

To obtain the total time complexity we need to add to the above-mentioned expression the time needed for computing $I^{(0)}$ and $O^{(0)}$. Since $I^{(0)}$ is simply the collection of sites, there are no distance comparisons done for computing it (only trivial memory operations which are $O\bigl(\sum_{k=1}^{|K|}|P_k|)\bigr)$. Since $O^{(0)}$ is the collection of Voronoi cells of the sites, since each cell is computed by computing endpoints along rays as explained above, since there are $C_1C_2^{\rm dim-1}|P_k|$ rays emanating from each $P_k$ and along each ray we perform $C_3A_k$ distance comparisons, and since in iteration $n=0$ each $A_k$ is represented by the points of $\cup_{j\neq k}P_j$,  the number of distance comparisons needed for calculating $O^{(0)}$ is $O\bigl(\sum_{k=1}^{|K|}\bigl(C_1C_2^{\rm dim-1}|P_k|C_3\sum_{j\neq k}|P_j|\bigr)\bigr)$. To conclude, the total time complexity for computing $I^{(0)}$,$O^{(0)},\ldots, I^{(n)}$, $O^{(n)}$ using the computational model and the ray-shooting algorithm mentioned in Subsections \bref{subsec:ComputationalModel}--\bref{subsec:InOn} above is
\begin{align*}
O\Biggl(\sum_{k=1}^{|K|}|P_k|\Biggr)&+O\Biggl(\sum_{k=1}^{|K|}\biggl(C_1C_2^{\rm dim-1}|P_k|C_3\sum_{j\neq k}|P_j|\biggr)\Biggr)\\
&+O\Biggl(2n\cdot C_3(C_1C_2^{\rm dim-1})^2\sum_{k=1}^{|K|}|P_k|\sum_{j\neq k}|P_j|\Biggr)\\
&\qquad\qquad=O\bigl(|K|G+C_3C_1C_2^{\rm dim-1}G^2(|K|-1)|K|(1+2nC_1C_2^{\rm dim-1})\bigr),
\end{align*}
where $G$ is any number satisfying $G\geq |P_k|$ for all $k\in K$. So essentially the time complexity grows quadratically with respect to the number of regions and number of points in each site, and at least exponentially with respect to the dimension (this growth is nothing but another evidence to the well-known curse of dimensionality phenomenon; other examples related to this phenomenon can be found in \cite[pp.\,93--94,\,99,\,135,\,141,\,282,\,284,\,286,\,297--298,\, 319]{Edelsbrunner-book-1987}, \cite[pp.\,205,\,212,\,231,\,233,\,234,\,236,\,241]{SharirAgarwal1995book}). Note that the numbers $|K|$, $|P_k|$, $k\in K$,  and $G$ (if used) are parameters depending on the input; $\rm dim$ and $C_1(\rm dim)$ are global  parameters; $C_2$ and $C_3$ are parameters depending on the user (the level of approximation). Better performance can be achieved by improved techniques for distance comparisons, in the spirit of \cite[Section 10]{ReemVoronoiParallel2015prep} and \cite[Section 2.4]{Reem2010PhD}. When such techniques are applied on certain configurations of sites, e.g., sites which are points which are uniformly distributed, then it is possible to reduce significantly the number of calculations. Another way to improve the performance is to apply parallel computing techniques. Such techniques are naturally supported in our setting since each region, and in fact, even each endpoint in each region, can be computed independently of the other ones.

\subsection{\bf A few clarifications}
We finish this section by clarifying certain issues which perhaps have not been very clear so far. It should be emphasized that the goal of this section  is to  describe schematically a
practical way for  approximating the regions which appear in each iteration and to roughly evaluate the number of calculations done in the process. While the description given here is not perfect, it is much more detailed than corresponding descriptions in previous works. Full analysis requires full analysis of the algorithm
presented in \cite{ReemISVD09} and this latter analysis is planned to be discussed elsewhere (done in a preliminary form in \cite{Reem2010PhD}). For instance, the fact that the algorithm presented in \cite{ReemISVD09} can approximate a given dominance region up to any desired precision is a consequence of the stability
of this algorithm whose proof is in the spirit of the one given in  \cite{ReemGeometricStabilityArxiv}.

If the users want to approximate (with respect to the Hausdorff distance) a given region up to some error parameter $\epsilon$, then they need to approximate well enough the sites, to choose in advance enough approximating rays (this is determined by the error parameter related
to the unit sphere), and to fix a small enough error parameter for the endpoints of the rays. If the users decide in advance how many iterations they want to perform for approximating the (double) zone diagram and what is the level of approximation of the regions in the final iteration, then they can estimate in advance the number of calculations by iterative ``reverse engineering'': using the target error parameter $\epsilon$, one estimates
the error parameters needed as input for the final iteration, and from them the error parameters needed for the previous iteration, and so on, until the initial iteration. This gives an estimate on the initial error parameters.

In the above description the number of iterations $n$ was chosen by the users but there was no guarantee that the real (double) zone diagram will be approximated well by $I^{(n)}$ or $O^{(n)}$. However, since the algorithm converges to the double zone diagram according to  Theorem \bref{thm:AlgUniConvex} above, then given $\epsilon>0$, there is a number $n_0$, depending only on $\epsilon$, such that for any integer $n\geq n_0$ the regions of $I^{(n)}$ and $O^{(n)}$ will be at Hausdorff  distance of at most $\epsilon$ from the limit regions. Hence, if one can compute $n_0(\epsilon)$, then one can know in advance the corresponding needed initial error parameters. Unfortunately, it is not clear how to estimate $n_0(\epsilon)$ and this is a major open problem.

\section{A byproduct: topological properties of Voronoi cells}\label{sec:TopologicalProperties}
The analysis (Section \bref{sec:proofs} below) leading to the convergence theorems mention in Section \bref{sec:ConvergenceResult} has resulted in several byproducts which we believe are of independent interest. In particular, it is shown in Theorem \bref{thm:BoundaryInterior}
below that under relatively general conditions the boundaries of the Voronoi cells
coincide with the corresponding bisectors, and hence these bisectors cannot be ``fat''.

\begin{thm}\label{thm:BoundaryInterior}
Let $(X,d)$ be a geodesic metric space which has the geodesic inclusion
property. Let $P,A\subseteq X$ be nonempty and disjoint. Suppose that for all $x\in X$
the distances $d(x,P)$ and $d(x,A)$ are attained  (namely, there exist $p\in P$ and $a\in A$ such that $d(x,P)=d(x,p)$ and $d(x,A)=d(x,a)$).
Let $\overline{S}$, $\Int(S)$, $\partial (S)$ be the closure/interior/boundary of $S\subseteq X$, respectively.
Then
\begin{equation}\label{eq:closure}
\dom(P,A)=\overline{\{x\in X: d(x,P)<d(x,A)\}},
\end{equation}
\begin{equation}
\partial(\dom(P,A))=\{x\in X: d(x,P)=d(x,A)\}=\bisect(P,A),\label{eq:boundary}
\end{equation}
\begin{equation}
\Int(\dom(P,A))=\{x\in X: d(x,P)<d(x,A)\},\label{eq:interior}
\end{equation}
and consequently, $\Int(\bisect(P,A))=\emptyset$, that is, the bisector between $P$ and $A$ cannot be ``fat''.
\end{thm}
Theorem \bref{thm:BoundaryInterior} generalizes a few results: \cite[p.\,111]{CMRS1993} (which refers to \cite{Mazon1992}) in which the setting is $\R^2$ with point sites and a strictly convex norm,   \cite[Theorem 2]{Wilker1975jour} (finite dimensional Euclidean spaces with sites having disjoint closures),  and \cite[Lemma 6]{ImaiKawamuraMatousekReemTokuyamaCGTA} (finite dimensional strictly convex normed spaces with disjoint closed sites). Its proof is inspired from a somewhat different  proof of
\cite[Lemma 6]{ImaiKawamuraMatousekReemTokuyamaCGTA} rather
than directly from \cite[Lemma 6]{ImaiKawamuraMatousekReemTokuyamaCGTA}. A careful analysis
of this second proof enabled us to generalize it to geodesic metric spaces having the geodesic inclusion property and actually to discover this class of spaces.

When the geodesic metric space does not have the geodesic inclusion property, then \beqref{eq:closure}--\beqref{eq:interior} can fail to hold and the bisector can be fat even in simple setting; see, e.g., \cite[p.\,390, Figure 37]{Aurenhammer}, \cite[p.\,191, Figure 3.7.2]{OBSC},
\cite[p.\,260 of the conference version]{ReemGeometricStabilityArxiv},\cite[Section 6]{ReemVorStabilityNonUC2012},\cite[Section 4]{ReemTopologicalPropertiesPreprint2013} (two or four point sites in the plane with either the $\ell_{1}$ or the $\ell_{\infty}$ norms). Recently results similar to Theorem \bref{thm:BoundaryInterior} have been established in
  \cite[Section 7]{ReemVorStabilityNonUC2012}, where the setting is a general
 normed space and one assumes that the distance to the sites is
attained and that the sites are aligned in a certain way with respect to the
structure of the unit sphere of the space (no two points of different sites form a line segment which is parallel to a nondegenerate line segment contained in the unit sphere; as before, when this alignment condition is not satisfied, then counterexamples exist, namely the same ones as above), and also in \cite{ReemTopologicalPropertiesPreprint2013}, where the setting is two positively separated sites of a general form and the space is a (possibly infinite dimensional) uniformly convex normed space.

We finish this short section by saying that there exist additional works which investigate  properties of bisectors in various settings. Some of them can be found in the following very partial list of references, as well as in some of the references cited therein: \cite{Busemann1955book,CorbalanMazonRecio1996,CMRS1993,Horvath2000,HorvathMartini2013jour,LevenSharir1987,Ma2000,MartiniSwanepoel2004,Mazon1992,PonceSantibanez2014jour}.

\section{Proofs of the convergence theorems and related claims}\label{sec:proofs}
This section presents the proofs of the convergence theorems and related claims.
We use the notation $\overline{S}$, $\Int(S)$, $\partial (S)$ to denote the closure/interior/boundary of $S\subseteq X$ respectively. We note that some of the assertions below hold in a more general setting than stated (for instance, in Lemma \bref{lem:GeneralDistance} the conditions on the distance function can be significantly weaken, so that one only needs, say,  the assumption that $d\colon X^2\to [-\infty,\infty]$), but we decided to restrict ourselves only to metric spaces (or particular classes of metric spaces).

The first lemmas describe simple properties of dom and Dom. Here and elsewhere we make use of the known and simple fact \cite[Lemma 5.4]{ReemReichZone} (in the Euclidean plane with point sites it was observed in \cite[Lemma 3(ii)]{AMT2}) that given a metric space $(X,d)$ and a fixed $\emptyset \neq P\subseteq X$, we have $\dom(P,B)\subseteq\dom(P,A)$ whenever $\emptyset\neq A\subseteq B$ (namely, $\dom(P,\cdot)$ is antimonotone).

\begin{lem}\label{lem:GeneralDistance}
Let $(X,d)$ be a metric space. Then
\begin{enumerate}[(a)]
\item\label{item:d_x_Agamma} $d\bigl(x,\bigcup_{\gamma\in \Gamma}A_{\gamma}\bigr)=\inf\{d(x,A_{\gamma}): \gamma\in \Gamma\}$ for any $x\in X$ and any collection $\{A_{\gamma}\}_{\gamma\in \Gamma}$ of nonempty subsets in $X$.
\item\label{item:UnionIntersection} $\dom\bigl(P,\bigcup_{\gamma\in \Gamma} A_{\gamma}\bigr)=\bigcap_{\gamma\in \Gamma} \dom(P,A_{\gamma})$
for any collection $\{A_{\gamma}\}_{\gamma\in \Gamma}$ of nonempty subsets in $X$ and any $P\subseteq X$ nonempty.
\item\label{item:FirstComponent} $\dom\bigl(\bigcup_{\gamma\in \Gamma} P_{\gamma},A\bigr)\subseteq\bigcup_{\gamma\in \Gamma} \dom(P_{\gamma},A)$
for any collection $\{P_{\gamma}\}_{\gamma\in \Gamma}$ of nonempty subsets in $X$ and any $A\subseteq X$ nonempty. If, in addition, for each $x\in X$ the distance $d\bigl(x,\bigcup_{\gamma\in \Gamma}P_{\gamma}\bigr)$ is attained at some $P_{\gamma_0}$, i.e., $d\bigl(x,\bigcup_{\gamma\in \Gamma}P_{\gamma}\bigr)=d(x,P_{\gamma_0})$ for some $\gamma_0\in \Gamma$, then the equality $\dom\bigl(\bigcup_{\gamma\in \Gamma} P_{\gamma},A\bigr)=\bigcup_{\gamma\in \Gamma} \dom(P_{\gamma},A)$ holds.
\item \label{item:DomUnionIntersect}$\Dom\bigl(\bigcup_{\gamma\in \Gamma} R^{\gamma}\bigr)=\bigcap_{\gamma\in\Gamma} \Dom(R^{\gamma})$ for any collection $\{R^{\gamma}\}_{\gamma\in \Gamma}$
of tuples, where each $R^{\gamma}$ is indexed by the same set of
indices $K$, i.e., $R^{\gamma}=(R_k^{\gamma})_{k\in K}$. Here $\Dom$ is defined with respect  to some tuple $(P_k)_{k \in K}$ of nonempty subsets in $X$.
\end{enumerate}
\end{lem}

\begin{proof}
\begin{enumerate}[(a)]
\item Let $x\in X$, $s:=d\bigl(x,\bigcup_{\gamma\in \Gamma}A_{\gamma}\bigr)$ and $t:=\inf\{d(x,A_{\gamma}): \gamma\in \Gamma\}$. Then $s\leq d(x,A_{\gamma})$ for all $\gamma\in \Gamma$ by the definition of $s$, and so $s\leq t$.  If $s<t$, then there is $y\in \bigcup_{\gamma\in \Gamma}A_{\gamma}$ such that $d(x,y)<t$, and since $y\in A_{\gamma}$ for some $\gamma\in \Gamma$, we have $d(x,A_{\gamma})\leq d(x,y)<t$, a contradiction with the definition of $t$. Hence $s=t$.
\item Since $\dom(P,\cdot)$ is antimonotone, we have $\textnormal{dom}\bigl(P,\bigcup_{i\in \Gamma} A_{i}\bigr)\subseteq
\textnormal{dom}(P,A_{\gamma})$ for any ${\gamma\in \Gamma}$. Consequently, we have
$\textnormal{dom}\bigl(P,\bigcup_{\gamma\in \Gamma} A_{\gamma}\bigr)\subseteq
\bigcap_{\gamma\in \Gamma}\textnormal{dom}(P,A_{\gamma})$. Conversely,
suppose that $x\in
\bigcap_{\gamma\in \Gamma}\textnormal{dom}(P,A_{\gamma})$. If $x\notin
\textnormal{dom}\bigl(P,\bigcup_{\gamma\in \Gamma} A_{\gamma}\bigr)$, then there is $y\in
\bigcup_{\gamma\in \Gamma} A_{\gamma}$ such that $d(x,P)>d(x,y)$. But $y\in
A_{\gamma}$ for some ${\gamma}$, and $x\in
\textnormal{dom}(P,A_{\gamma})$, so $d(x,P)\leq d(x,A_{\gamma})\leq
d(x,y)$, a contradiction. Thus $x\in
\textnormal{dom}\bigl(P,\bigcup_{\gamma\in \Gamma} A_{\gamma}\bigr)$.
\item By monotonicity of $\dom(\cdot, A)$, i.e., $\dom(P,A)\leq \dom(Q,A)$ whenever $\emptyset\neq P\subseteq Q$ (which is easily verified), we have $\dom(P_{\gamma},A)\subseteq \dom(\bigcup_{i\in \Gamma} P_{i},A)$ for any ${\gamma}\in \Gamma$. As a result,
$\bigcup_{\gamma\in \Gamma}\dom(P_{\gamma},A)\subseteq \dom\bigl(\bigcup_{\gamma\in \Gamma}
P_{\gamma},A\bigr)$. Now fix $x\in\dom\bigl(\bigcup_{\gamma\in \Gamma} P_{\gamma},A\bigr)$. We know that $d\bigl(x,\bigcup_{\gamma\in \Gamma}P_{\gamma}\bigr)=d(x,P_{\gamma_0})$ for some $\gamma_0\in \Gamma$. Therefore,  $d(x,P_{\gamma_0})=d\bigl(x,\bigcup_{\gamma\in \Gamma}P_{\gamma}\bigr)\leq d(x,A)$. Consequently, $x\in  \dom(P_{\gamma_0},A)\subseteq\bigcup_{\gamma\in \Gamma}\dom(P_{\gamma},A)$.
\item Since the union and intersection are taken component-wise, the assertion follows from the definition of the $\Dom$ mapping (equation \beqref{eq:TZoneDef})
and part (\bref{item:UnionIntersection}):
\begin{align*}
\quad \Dom\biggl(\bigcup_{\gamma\in
\Gamma}R^{\gamma}\biggr)&=\biggl(\textnormal{dom}\biggl(P_k,\bigcup_{j\neq
k}\biggl(\bigcup_{\gamma\in \Gamma}{R_j^{\gamma}}\biggr)\biggr)\biggr)_{k\in
K}=\biggl(\textnormal{dom}\biggl(P_k,\bigcup_{\gamma\in \Gamma}\biggl(\bigcup_{j\neq
k}{R_j^{\gamma}}\biggr)\biggr)\biggr)_{k\in
K}\\&=\biggl(\bigcap_{\gamma\in \Gamma}\textnormal{dom}\biggl(P_k,\bigcup_{j\neq
k}{R_j^{\gamma}}\biggr)\biggr)_{k\in K}=\bigcap_{\gamma\in \Gamma}\Dom(R^{\gamma}).
\end{align*}
\end{enumerate}
\end{proof}

\begin{lem}\label{lem:DomPk}
Let $(P_k)_{k \in K}$ be a tuple of  nonempty and closed subsets in a metric space $(X,d)$.
\begin{enumerate}[(a)]
\item $\Dom$ is antimonotone, i.e., $\Dom(S)\subseteq \Dom(R)$ whenever $R,S\in \underset{{k\in K}}\prod
X_k$ (see Definition \bref{def:zone}) and $R\subseteq S$. In addition, $\Dom^2$ is monotone, that is, $R\subseteq S\Rightarrow \Dom^2(R)\subseteq \Dom^2(S)$.

\item\label{item:In R On}  Consider the inner approximation sequence $(I^{(n)})_{n=0}^{\infty}$ and the outer one $(O^{(n)})_{n=0}^{\infty}$ (Algorithm  \bref{alg:InOn} above). Then $I^{(n)}=\Dom^{2n}(P_k)_{k\in K}$ and $O^{(n)}=\Dom^{2n+1}(P_k)_{k\in K}$ for each $n\in\N\cup\{0\}$, the sequence $(I^{(n)})_{n=0}^{\infty}$ increasing, the sequence $(O^{(n)})_{n=0}^{\infty}$ is decreasing,
and each nonnegative integer $n$
\begin{equation}\label{eq:InOn}
I^{(n)}\subseteq \bigcup_{q=0}^{\infty}I^{(q)}\subseteq \bigcap_{q=0}^{\infty}O^{(q)}\subseteq O^{(n)}.
\end{equation}
In addition, if $R$ is a zone or double zone diagram in $X$, then for each $n\in\N\cup\{0\}$
\begin{equation}\label{eq:DomIterations}
I^{(n)}\subseteq\bigcup_{q=0}^{\infty}I^{(q)}\subseteq R \subseteq \bigcap_{q=0}^{\infty}O^{(q)}\subseteq O^{(n)}.
\end{equation}
\end{enumerate}
\end{lem}
\begin{proof}
\begin{enumerate}[(a)]
\item This claim is known: see, e.g., \cite[Lemma 5.4]{ReemReichZone} (in the Euclidean plane it was observed in \cite[p.\,1184]{AMTn}).
\item The definition of $I^{(n)}$ and $O^{(n)}$ (Algorithm \bref{alg:InOn} above) and induction show that $I^{(n)}=\Dom^{2n}(P_k)_{k\in K}$ and $O^{(n)}=\Dom^{2n+1}(P_k)_{k\in K}$ for each $n\in\N\cup\{0\}$. Now let $(X)_{k\in K}$ be the tuple whose components are $X$. An immediate check shows that $P_k=\dom(P_k,X)$ and $\Dom(X)_{k\in K}=(P_k)_{k \in K}$. This equality, the inclusion $(P_k)_{k\in K}\subseteq (X)_{k\in K}$, the definition of the set on which $\Dom$ acts (Definition \bref{def:zone} above), and the monotonicity of $\Dom^2$, imply that
\begin{equation*}
I^{(0)}=(P_k)_{k\in K}\subseteq \Dom^2 (P_k)_{k\in K}\subseteq \Dom^2 (X)_{k\in K}=\Dom(P_k)_{k\in K}=O^{(0)}\subseteq (X)_{k\in K}.
\end{equation*}
By iterating this inclusion with $\Dom^2$, using the monotonicity of $\Dom^2$, and using the equalities $I^{(n)}=\Dom^{2n}(P_k)_{k\in K}$ and $O^{(n)}=\Dom^{2n+1}(P_k)_{k\in K}$ for each $n\in\N\cup\{0\}$, we see that $(I^{(n)})_{n=0}^{\infty}$ is increasing, $(O^{(n)})_{n=0}^{\infty}$ is decreasing, and \beqref{eq:InOn} holds.
Finally, after iterating the inclusion $(P_k)_{k\in K}\subseteq R\subseteq (X)_{k\in K}$ using $\Dom^2$ we obtain  \beqref{eq:DomIterations} when $R$ is a double zone diagram. But a zone diagram is a double  zone diagram
[$\Dom(\Dom(R))=\Dom(R)=R$ whenever $R=\Dom(R)$], hence  \beqref{eq:DomIterations} is true also when $R$ is a zone diagram.
\end{enumerate}
\end{proof}

\begin{lem}\label{lem:dom}
Let $(X,d)$ be a metric space and let $P,A\subseteq X$ be nonempty.  Then
\begin{enumerate}[(a)]
\item\label{item:closure} $\dom(P,A)$ is a closed set.
\item\label{item:Aclosure} $\dom(P,A)=\dom(P,\overline{A})=\dom(\overline{P},A)=\dom(\overline{P},\overline{A})$.
\item\label{item:DomClosure} For a tuple $R=(R_k)_{k\in K}$ of nonempty subsets, let $\overline{R}=(\overline{R_k})_{k\in K}$.
Then $\Dom(\overline{R})=\Dom(R)=\overline{\Dom(R)}$ (with respect to a given  tuple $(P_k)_{k\in K}$).
\end{enumerate}
\end{lem}
\begin{proof}
\begin{enumerate}[(a)]
\item If $x=\lim_{n\to \infty}x_n$ where $x_n\in \dom(P,A)$ for all $n\in\N$, then $d(x_n,P)\leq
d(x_n,A)$, and this inequality is preserved in the limit because the
function $x\mapsto d(x,P)-d(x,A)$ is continuous (with respect to the topology
induced on $X$ by $d$). Thus the limit of a sequence in $\dom(P,A)$ is also in $\dom(P,A)$, namely $\dom(P,A)$ is closed.
\item This claim follows from the general fact that $d(x,B)=d(x,\overline{B})$ for each  $x\in X$ and each $B\subseteq X$.
\item By part \beqref{item:Aclosure} and by Lemma \bref{lem:GeneralDistance}\,\beqref{item:UnionIntersection}, we have
\begin{equation*}
\quad\quad \dom\biggl(P_k,\bigcup_{j\neq k}\overline{R_j}\biggr)=\bigcap_{j\neq k}\dom(P_k,\overline{R_j})
=\bigcap_{j\neq k}\dom(P_k,R_j)=\dom\biggl(P_k,\bigcup_{j\neq k}R_j\biggr).
\end{equation*}
Thus $\Dom(\overline{R})=\bigl(\dom\bigl(P_k,\bigcup_{j\neq k}\overline{R_j}\bigr)\bigr)_{k\in K}=\bigl(\dom\bigl(P_k,\bigcup_{j\neq k}R_j\bigr)\bigr)_{k\in K}=\Dom(R)$. Finally, by part \beqref{item:closure} and by definition, we have
\begin{equation*}
\overline{\Dom(R)}=\biggl(\overline{\dom\biggl(P_k,\bigcup_{j\neq k}P_j\biggr)}\biggr)_{k\in K}=
\biggl(\dom\biggl(P_k,\bigcup_{j\neq k}P_j\biggr)\biggr)_{k\in K}=\Dom(R).
\end{equation*}
\end{enumerate}
\end{proof}

The following lemma and propositions (followed by some proofs
of related claims) discuss issues related to geodesic metric spaces.

\begin{lem}\label{lem:GeodesicEquality}
Let $(X,d)$ be a geodesic metric space and let $a,b,c\in X$.
Then equality in the triangle inequality $d(a,c)\leq d(a,b)+d(b,c)$ holds
if and only if  $b\in [a,c]_{\gamma}$ for some isometric mapping $\gamma$.
\end{lem}
\begin{proof}
Suppose first that $b\in [a,c]_{\gamma}$ for some $b\in X$ and some
isometric mapping $\gamma: [r_1,r_3]\to X$. We can write $a=\gamma(r_1),b=\gamma(r_2)$, and $c=\gamma(r_3)$
for some real number $r_2\in [r_1,r_3]$. Since $\gamma\colon[r_1,r_3]\to X$ preserves
distances we have
\begin{align*}
d(a,b)+d(b,c)&=d(\gamma(r_1),\gamma(r_2))+d(\gamma(r_2),\gamma(r_3))\\
&=(r_2-r_1)+(r_3-r_2)=r_3-r_1=d(a,c),
\end{align*}
namely equality in the triangle inequality.

Assume now the reverse direction, i.e., that $b\in X$ satisfies $d(a,b)+d(b,c)=d(a,c)$.
Since $(X,d)$ is a geodesic metric space we can find isometric mappings  $\gamma_1\colon[r_1,r_2]\to X$
and $\gamma_2\colon[r_2,r_3]\to X$ such that $\gamma_1(r_1)=a$, $\gamma_1(r_2)=b$, $d(a,b)=r_2-r_1$,
 $\gamma_2(r_2)=b$, $\gamma_2(r_3)=c$, and $d(b,c)=r_3-r_2$ (for some real numbers $r_1\leq r_2\leq r_3$).
Define $\gamma\colon[r_1,r_3]\to X$ by
\begin{equation*}
\gamma(r)=\left\{ \begin{array}{ll}
\gamma_1(r) & \textnormal{if}\,\, r\in [r_1,r_2],\\
\gamma_2(r) & \textnormal{if}\,\, r\in [r_2,r_3].
\end{array}
\right.
\end{equation*}
It will be proved that $\gamma$ is an isometric mapping and that $b\in [a,c]_{\gamma}$.
By the definition of $\gamma$ we have $\gamma(r_2)=b$ and hence $b\in \gamma([r_1,r_3])$.
For proving that $\gamma$ preserves
distances, one observes that when restricted to $[r_1,r_2]$ or to $[r_2,r_3]$ this holds
because $\gamma_i,i=1,2$ preserve distances. It remains to show that if
 $t_1\in [r_1,r_2]$ and $t_2\in [r_2,r_3]$, then the equality $d(\gamma(t_1),\gamma(t_2))=t_2-t_1$ holds.
 Indeed, by the triangle inequality, by the fact that the restriction of $\gamma$
 to the intervals $[r_1,r_2]$ and $[r_2,r_3]$ is an isometric mapping,
and using the properties of $a,b,c$ that were assumed above, it follows that
\begin{equation}\label{eq:EqualityTriangle}\begin{split}
d(a,c)&\leq d(a,\gamma(t_1))+d(\gamma(t_1),\gamma(t_2))+d(\gamma(t_2),c)\\
&\leq d(\gamma(r_1),\gamma(t_1))+d(\gamma(t_1),\gamma(r_2))+d(\gamma(r_2), \gamma(t_2))+d(\gamma(t_2),\gamma(r_3))\\
&=(t_1-r_1)+(r_2-t_1)+(t_2-r_2)+(r_3-t_2)=(r_2-r_1)+(r_3-r_2)\\
&=d(a,b)+d(b,c)=d(a,c).
\end{split}
\end{equation}
Thus there is equality all over the way. This fact and
the fact that $\gamma$ preserves distances when restricted to the intervals $[r_1,r_2]$ and $[r_2,r_3]$ imply
the desired equality
\begin{equation}\label{eq:t_1t_2}
d(\gamma(t_1),\gamma(t_2))=d(\gamma(t_1), \gamma(r_2))+d(\gamma(r_2), \gamma(t_2))=(r_2-t_1)+(t_2-r_2)=t_2-t_1.
\end{equation}
\end{proof}

\begin{prop}\label{prop:GeodesicInclusionStrictlyConvex}
Let $(X,|\cdot|)$ be a normed space. Then $(X,|\cdot|)$ is strictly convex if and only if it has the geodesic inclusion property ( Definition \bref{def:GeodesicMetric}\,\beqref{def:GeodesicMetric:GeodesicInclusion} above). Moreover, given a convex subset of $C$ of $X$ having at least two different points, if $(X,|\cdot|)$ is strictly convex, then $C$ with the induced norm from $X$ has the geodesic inclusion property.
\end{prop}
\begin{proof}
Assume first that $(X,|\cdot|)$ is strictly convex and let $C$ be a convex subset of $X$ having at least two different points. Given two different points $x,y\in C$, we claim that a geodesic segment
$[x,y]_{\gamma}$ which connects $x=\gamma(r_1)$ and $y=\gamma(r_2)$ in $X$ (where $r_2> r_1$
and $\gamma\colon[r_1,r_2]\to X$ preserves distances) must be the line segment $[x,y]$ (in particular, any geodesic segment which connects $x$ and $y$ in $C$, i.e., $\gamma([r_1,r_2])\subseteq C$, must coincide with $[x,y]$).
Indeed, let $b:=\gamma(r)\in [x,y]_{\gamma}$  for some $r\in (r_1,r_2)$. Then
\begin{align*}
|x-y|\leq |x-b|+|b-y|&=|\gamma(r_1)-\gamma(r)|+|\gamma(r)-\gamma(r_2)|\\
&=r-r_1+r_2-r=r_2-r_1=|x-y|,
\end{align*}
showing that $|x-b|+|b-y|=|x-y|$. This is in contrast to the assumed strict convexity of the space
unless  $b$ belongs to the line segment $[x,y]$. Thus $[x,y]_{\gamma}\subseteq [x,y]$.
On the other hand, since $\gamma$ is continuous, its image $\gamma([r_1,r_2])=[x,y]_{\gamma}$
is a (nontrivial) connected subset of $[x,y]$. But the only nontrivial connected subsets of a line segment
are line segments. We conclude that $[x,y]_{\gamma}$ is a line segment which contains
$x$ and $y$. The equality $[x,y]_{\gamma}=[x,y]$ follows.

\begin{figure}
\begin{minipage}[t]{0.45\textwidth}
\begin{center}
{\includegraphics[scale=1]{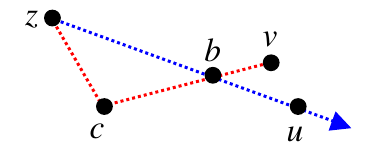}}
\end{center}
 \caption{Illustration of the proof of Proposition \bref{prop:GeodesicInclusionStrictlyConvex} in the non-strictly convex case.}
\label{fig:NotStrictlyConvex}
\end{minipage}
\end{figure}

Returning to the geodesic inclusion property that we want to prove that $C$ has, let $u,v,b,z\in C$ satisfy
$b\neq z$, $b\in [u,z]_{\gamma_1}$, and $b\in [v,z]_{\gamma_2}$. The geodesic inclusion property obviously holds if $u=b$ or $v=b$.
Otherwise $|u-b|>0$, $|v-b|>0$.  Since we already know from the previous paragraph that the geodesic segments in $C$ are line segments, we have $b\in [u,z]$ and $b\in [v,z]$. Thus $(b-z)/|b-z|=(u-b)/|u-b|$ and $(b-z)/|b-z|=(v-b)/|v-b|$,
showing that $[b,u]$ and $[b,v]$ are on the same ray emanating from $b$ in the joint
direction $(b-z)/|b-z|$. Hence the shorter interval among $[b,u]$, $[b,v]$
is contained in the longer one and it follows that either $u\in [v,z]$ or $v\in [u,z]$. Since any line segment is a geodesic segment, we conclude that $C$ with the induced metric from the norm has the geodesic inclusion property, as required.

Now assume that the given normed space $(X,|\cdot|)$ is not strictly convex. Our goal is to prove that this space does not have the geodesic inclusion property. See Figure \bref{fig:NotStrictlyConvex} for an illustration of the proof. First we observe that each such a space must have dimension of at least 2 (since any norm $\|\cdot\|$ on a one-dimensional normed space must have the form $\|x\|=\lambda|x|$ for all $x\in X$ where $\lambda>0$ is a given positive constant not depending on $x$ and where $|\cdot|$ is the usual absolute value function; hence this space is strictly convex). Since the space is not strictly convex, we can find three different points $v,c,z$ such that $c\notin [v,z]$ and still there is equality in the triangle inequality, that is, $|z-v|=|z-c|+|c-v|$. Therefore, if we define a function $\gamma$ from the real-line interval $[0,|z-v|]$ to $X$ by $\gamma(t):=v+t(c-v)/|c-v|$ when $t\in [0,|c-v|]$ and $\gamma(t):=c+t(z-c)/|z-c|$ when $t\in [|c-v|,|z-v|]$ , then it can be checked similarly to the proof of Lemma \bref{lem:GeodesicEquality} above (by minor variations of \beqref{eq:EqualityTriangle} and \beqref{eq:t_1t_2}) that $\gamma$ is an isometry.  Since $\Gamma_1:=[v,c]\cup [c,z]=\gamma([0,|z-v|])$, it follows that $\Gamma_1$ is a geodesic segment. Now we select a point $b\in (v,c)$ and a point $u\notin [b,z]$ but on the ray emanating from $z$ in the direction of $b$. The set $\Gamma_2:=[u,z]$ is a line segment and hence a geodesic segment. It follows that $b\in \Gamma_1\cap \Gamma_2$ but neither $u\in\Gamma_1$ nor $v\in \Gamma_2$. Hence $X$ does not have the geodesic inclusion property, as claimed.
\end{proof}

\begin{prop}\label{prop:StrictSegment}
Let $(X,d)$ be a geodesic metric space which has the geodesic
inclusion property.  Let $A\subseteq X$ be nonempty.
Let $p,z\in X$ and suppose that $d(z,p)\leq d(z,A)$ and $p\notin A$.
Suppose also that $d(x,A)$ is attained for all $x\in [p,z)$.
Then $d(x,p)<d(x,A)$ for all $x\in [p,z)$. In particular, if $(X,d)$ is proper,
 $\emptyset\neq A$ is closed and for some $p,z\in X$ we have $d(z,p)\leq d(z,A)$ and $p\notin A$,
then $d(x,p)<d(x,A)$ for all $x\in [p,z)$.
\end{prop}
\begin{proof}
It can be assumed that  $z\neq p$, because otherwise the assertion is obvious (void).
Fix $x\in [p,z)$ and let $a_x\in A$ be chosen such
that $d(x,A)=d(x,a_x)$.
Since $p\notin A$, $z\neq p$, and $d(z,p)\leq d(z,A)$, it follows that $z,p$ and $a_x$ are all different.
Using the facts that $[p,z]$ is a geodesic segment, that $x\in [p,z]$, and that $d(z,p)\leq d(z,A)$, it follows
from the triangle inequality and Lemma \bref{lem:GeodesicEquality} that
\begin{equation}\label{eq:zxpA}
\quad \quad d(z,x)+d(x,p)=d(z,p)\leq d(z,A)
\leq d(z,a_x)\leq d(z,x)+d(x,a_x),
\end{equation}
so $d(x,p)\leq d(x,a_x)=d(x,A)$. Assume for a contradiction that
\begin{equation}\label{eq:d(x,p)d(x,a_x)}
d(x,p)=d(x,a_x).
\end{equation}
Then there is equality in \beqref{eq:zxpA} and hence, because of Lemma \bref{lem:GeodesicEquality},
it follows that $x\in [a_x,z]_{\gamma_1}$ for some isometric mapping $\gamma_1$.
Because $x\neq z$ the inclusion $x\in [a_x,z)_{\gamma_1}$ holds.
However, we have already assumed that $x\in [p,z)$ and $[p,z]$ is a geodesic segment, namely $x\in [p,z)_{\gamma_2}$ for some isometric mapping $\gamma_2$.
Thus, using the geodesic inclusion property of the space
(where $x$ is the bifurcation point), either $p\in [a_x,z]_{\gamma_1}$
or $a_x\in [p,z]_{\gamma_2}$. As explained after the definition of
the geodesic inclusion property, in the first case actually $p\in[a_x,x]_{\gamma_1}$ (otherwise
$p\in (x,z]_{\gamma_1}$ and hence, since $\gamma_1$ preserves distances, $d(p,z)<d(x,z)$;
but $d(x,z)\leq d(p,z)$ because $x\in [p,z]_{\gamma_2}$, a contradiction). Consequently, the relations $p\in[a_x,x]_{\gamma_1}$ and $p\neq a_x$ imply that
$d(x,a_x)=d(x,p)+d(p,a_x)>d(x,p)$, a contradiction to \beqref{eq:d(x,p)d(x,a_x)}.
Thus the second case implied by the geodesic inclusion property holds, i.e., $a_x\in [p,z]_{\gamma_2}$.
But then actually $a_x\in [p,x]_{\gamma_2}$ by a similar argument as used above in the first case.
Therefore $d(p,x)=d(p,a_x)+d(a_x,x)$, and so \beqref{eq:d(x,p)d(x,a_x)} implies that $d(p,a_x)=0$,
a contradiction to the fact established earlier that $p\neq a_x$. Therefore \beqref{eq:d(x,p)d(x,a_x)} cannot hold, that is, $d(x,p)<d(x,A)$, as required.

Finally, if $(X,d)$ is proper, $\emptyset\neq A$ is closed and for some $p,z\in X$ we have $d(z,p)\leq d(z,A)$ and $p\notin A$,
then a simple well-known argument based on compactness (without referring to any geodesic assumption on $X$)
shows that $d(x,A)$ is actually attained for all $x\in X$ and hence the previous paragraphs imply that
$d(x,p)<d(x,A)$ for all $x\in [p,z)$.
\end{proof}
For a later use, it is the time to prove  Theorem \bref{thm:BoundaryInterior}.
\begin{proof}[{\bf Proof of Theorem \bref{thm:BoundaryInterior}}]
We first establish \beqref{eq:closure}.
From the inclusion $\{x\in X: d(x,P)<d(x,A)\}\subseteq \dom(P,A)$ and the fact that $\dom(P,A)$
is closed (as proved in Lemma \bref{lem:dom}\,\beqref{item:closure}) it follows that the inclusion $\overline{\{x\in X: d(x,P)<d(x,A)\}}\subseteq \dom(P,A)$
holds. For the second inclusion, let $z\in \dom(P,A)$. If $d(z,P)<d(z,A)$, then obviously the point $z$ is in the set $\overline{\{x\in X: d(x,P)<d(x,A)\}}$. Now suppose that $d(z,P)=d(z,A)$.
By assumption there is $p\in P$ such that $d(z,P)=d(z,p)$, and since $p\in P$, it  follows that $d(p,P)=0$ and $p\notin A$. These properties and the fact that $d(p,A)$ is attained imply that $d(p,A)>0$. Hence if $z=p$, then $z\in \overline{\{x\in X: d(x,P)<d(x,A)\}}$.  If $z\neq p$, then $[p,z)\neq \emptyset$, and Proposition  \bref{prop:StrictSegment} implies that every $x\in [p,z)$ (arbitrarily  close to $z$) satisfies the inequality $d(x,P)\leq d(x,p)<d(x,A)$. Thus again $z\in \overline{\{x\in X: d(x,P)<d(x,A)\}}$.

We now turn to prove \beqref{eq:boundary}. Let $f:X\to \R$ be defined by  $f(x):=d(x,P)-d(x,A)$ for all $x\in X$.
Because of the continuity of $f$ and because $\dom(P,A)=\{x\in X: f(x)\leq 0\}$
one obtains that all the points
in $\{x\in X: d(x,P)<d(x,A)\}$ and $\{x\in X: d(x,A)<d(x,P)\}$
are interior points. It follows that $\partial(\dom(P,A))\subseteq \{x\in X: d(x,P)=d(x,A)\}$
without any assumption on the sites or on $X$. For the reverse inclusion, let
 $z$ satisfy $d(z,P)=d(z,A)$. Then $z\in \dom(P,A)$. Let $a\in A$ satisfy $d(z,a)=d(z,A)$.
 Then $a\neq z$ since otherwise the equality $d(z,P)=d(z,A)$ and the fact that $d(z,P)$ is
 attained would imply that $z\in P\cap A$, a contradiction. Hence $[a,z)\neq\emptyset$. The
 inclusion  $[a,z)\subseteq \{x\in X: d(x,A)<d(x,P)\}$ holds
because of Proposition \bref{prop:StrictSegment} (with $P$ instead of $A$
and $a$ instead of $p$) and it proves that arbitrarily close to $z$
there are points outside $\dom(P,A)$. Thus 	\beqref{eq:boundary} holds.

To deduce \beqref{eq:interior} we use \beqref{eq:boundary}, the equality
\begin{multline*}
\partial(\dom(P,A))\cup \Int(\dom(P,A))=\dom(P,A)\\= \{x\in X: d(x,P)=d(x,A)\}\cup\{x\in X: d(x,P)<d(x,A)\},
\end{multline*}
 and the fact that the terms in both unions are disjoint. Finally, to see that the equality $\Int(\bisect(P,A))=\emptyset$ holds, assume for a contradiction that this equality does not hold. Then there is a point $x_0\in X$ which is the center of an open ball $B$ contained in $\bisect(P,A))$. Since we know from \beqref{eq:boundary} that  $\bisect(P,A)=\partial(\dom(P,A))$, we conclude that there are in $B$ points from $\dom(P,A)$ and from $X\backslash \dom(P,A)$. Let $x\in B\cap (X\backslash \dom(P,A))$. Then, on the one hand, $d(x,P)=d(x,A)$ (since $x\in B\subseteq\bisect(P,A)$), but on the other hand  $d(x,A)<d(x,P)$ (since $x\notin \dom(P,A)$), a contradiction which proves the assertion.
\end{proof}

The next two lemmas have a somewhat independent nature but they are needed for a later use.
They are  probably known, and their proofs are given for the sake of completeness.
\begin{lem}\label{lem:DistIntersect}
Let $(X,d)$ be a proper metric space and let $(A_{\gamma})_{\gamma=1}^{\infty}$ be a decreasing sequence
of closed subsets of $X$ such that $A:=\bigcap_{\gamma=1}^{\infty}A_{\gamma}$ is nonempty.
Then $d(x,A)=\lim_{\gamma\to\infty}d(x,A_{\gamma})$  for each $x\in X$.
\end{lem}
\begin{proof}
Let $x\in X$. The definition of a proper metric space implies that any bounded sequence in $X$ has a convergent
subsequence.  Thus, a well known and simple argument shows that the distance between $x$ and any nonempty closed subset is attained. Therefore for each $\gamma\in\N$ there exists $x_{\gamma}\in A_{\gamma}$ such that
$d(x,A_{\gamma})=d(x,x_{\gamma})$. Since $(A_{\gamma})_{\gamma=1}^{\infty}$ is decreasing and $A_{\gamma}$ is closed
for each $\gamma$, it follows that any limit point of the sequence $(x_{\gamma})_{\gamma\in \N}$ is in $A_{\gamma}$ for each $\gamma$ and hence in  $A$. Since
$(A_{\gamma})_{\gamma=1}^{\infty}$ is decreasing, the sequence $(d(x,x_{\gamma}))_{\gamma\in \N}$ is increasing and its limit is bounded by $d(x,A)$. Hence $(x_{\gamma})_{\gamma\in \N}$ is bounded and  there exists $x_{\infty}\in A$ and a subsequence such that $x_{\infty}=\lim_{\beta\to\infty}x_{\gamma_{\beta}}$. From the continuity of the distance function we deduce that $d(x,x_{\infty})=\lim_{\beta\to\infty}d(x,x_{\gamma_{\beta}})\leq d(x,A)$. But
$d(x,A)\leq d(x,x_{\infty})$ by definition because $x_{\infty}\in A$. Therefore  $\lim_{\gamma\to\infty}d(x,A_{\gamma})=\lim_{\beta\to\infty}d(x,A_{\gamma_{\beta}})=d(x,x_{\infty})=d(x,A)$.
\end{proof}

\begin{lem}\label{lem:HausdorffUnionIntersection}
Let $(X,d)$ be a compact metric space. Let $(A_n)_{n=1}^{\infty}$ be a decreasing sequence  of nonempty closed sets of $X$, and let $(B_n)_{n=1}^{\infty}$ be an increasing  sequence of nonempty sets of $X$. Let $A:=\bigcap_{n=1}^{\infty}A_n$ and $B:=\overline{\bigcup_{n=1}^{\infty}B_n}$. Then $A$ and $B$ are nonempty and
$A=\lim_{n\to \infty}A_n$ and $B=\lim_{n\to\infty}B_n$ with respect to the Hausdorff distance.
\end{lem}
\begin{proof}
First, we observe that $A\neq\emptyset$ since the family $\{A_n: n\in\N\}$ is a family of closed sets having the finite intersection  property  and the space is compact. Now suppose by way of contradiction that it is not true that $\lim_{n\to\infty}D(A_n,A)=0$, where $D$ is the Hausdorff distance. Then there exists some $\epsilon>0$ and a subsequence $(n_m)_{m\in N}$ of natural numbers such that $\epsilon\leq D(A_{n_m},A)$ for each $m\in\N$. Since $A\subseteq A_n$ for each $n\in\N$, the definition of the Hausdorff  distance (Definition \bref{def:Hausdorff} above) and the fact that the function $x\mapsto d(x,A)$ is continuous on the closed and hence compact set $A_n$ for each $n\in\N$ (and thus attains a maximum value there), imply that for each $m\in\N$ there exists $x_m\in A_{n_m}$ such that $\epsilon \leq d(x_m,A)$. Let $(x_{m_l})_{l\in\N}$  be a convergent subsequence of the sequence $(x_m)_{m=1}^{\infty}$, and let $x$ be its limit. Since the sequence $(A_n)_{n\in\N}$ is decreasing, it follows that $x_{m_{l'}}\in A_{n_{m_l}}$ whenever $l,l'\in\N$ and  $l\leq l'$. Therefore $x\in A_{n_{m_l}}$ for all $l\in\N$ because $A_{n_{m_l}}$ is closed. Thus $x\in \bigcap_{l=1}^{\infty}A_{n_{m_l}}$.  But $\bigcap_{l=1}^{\infty}A_{n_{m_l}}=A$ because the intersection is decreasing. Hence  $0=d(x,A)=\lim_{l\to\infty}d(x_{m_l},A)\geq \epsilon$, a contradiction.

Now we turn to proving that $B=\lim_{n\to\infty}B_n$ (the statement that $B\neq\emptyset$ is obvious). Suppose by way of contradiction that this is not true. Then there exists some $\epsilon>0$ and a subsequence $(n_m)_{m=1}^{\infty}$ of natural numbers such that $\epsilon\leq D(B_{n_m},B)$ for each $m\in\N$. Since $B_n\subseteq B$ for each $n\in\N$, the definition of the Hausdorff  distance and the fact that for each $m\in\N$ the function $x\mapsto d(x,B_{n_m})$ is continuous on the closed and hence compact set $B$ for each $n\in\N$ (and thus attains a maximum value there) imply that there exists $x_m\in B$ such that $\epsilon \leq d(x_m,B_{n_m})$.  Let $(x_{m_l})_{l\in\N}$  be a convergent subsequence of the sequence $(x_m)_{m=1}^{\infty}$, and let $x$ be its limit. Then $d(x,x_{m_l})<\epsilon/4$ for all $l\in\N$ large enough. Since $B$ is closed it follows that $x\in B$, so by the definition of $B$ there exists $n_0\in \N$ such that $d(x,B_{n_0})<\epsilon/2$. But $B_{n_0}\subseteq B_n$ for each $n\in\N$ larger than $n_0$, and so $d(x,B_n)\leq d(x,B_{n_0})<\epsilon/2$ for all such $n$. In particular,  $d(x,B_{n_{m_l}})<\epsilon/2$ for $l\in\N$ sufficiently large. Thus
\begin{equation*}
\epsilon\leq d(x_{m_l},B_{n_{m_l}})\leq d(x_{m_l},x)+d(x,B_{n_{m_l}})<3\epsilon/4,
\end{equation*}
a contradiction.
\end{proof}

The next claims discuss several subtler properties of $\dom$ and $\Dom$.  Lemma \bref{lem:AntiIntersec_dom} and
Proposition \bref{prop:mDomM} generalize \cite[Lemma 3.1]{AMTn} and \cite[Lemma 5.1]{AMTn} respectively, and are partly inspired by them.
\begin{lem}\label{lem:Dom3}
Let $(X,d)$ be a metric space and let $P=(P_k)_{k\in K}$ be a tuple of nonempty subsets of $X$ having the property that for all $k\in K$,
\begin{equation*}
r_k:=\inf\{d(P_k,P_j):j\neq k\}>0.
\end{equation*}
For $Q\subseteq X$ nonempty and $r>0$, let $B(Q,r):=\{x\in X:d(x,Q)<r\}$. We denote by $(\Dom P)_k$ the $k$-th  component of $\Dom P$. Then:
\begin{enumerate}[(a)]
\item\label{item:B(P,r)Dom2} $B(P_k,r_k/4)\subseteq (\Dom^2 P)_k\subseteq (\Dom^{\gamma} P)_k$ for each integer $\gamma\geq 1$ and any $k\in K$.
\item \label{item:r} If $(X,d)$ is a geodesic metric space, then $d(x,B(Q,r))+r=d(x,Q)$ for each $Q\subseteq X$ nonempty, each $r>0$ and each $x\notin B(Q,r)$.
\item\label{item:Dom3Disjoint} If $d(x,B(P_k,r))<d(x,P_k)$ for all $k\in K$, $r>0$, and  $x\notin \overline{P_k}$, then
for each integer $\gamma\geq 2$ the components of $\Dom^{\gamma}P $ are disjoint. Moreover, if $X$ is a geodesic metric
 space, then for all $j,k\in K$,$j\neq k$ and all integer $\gamma\geq 2$, the following inequality holds: $(r_k/8)+(r_j/8)\leq d((\Dom^{\gamma}P)_k,(\Dom^{\gamma} P)_j)$.
\item\label{item:r_k} Suppose  that $(X,d)$ is a geodesic metric space. Given $j,k\in K$, $j\neq k$, and an integer $\gamma\geq 2$, if $x\in (\Dom^{\gamma}P)_j$, then $r_k/4\leq d(x,P_k)$.
\end{enumerate}
\end{lem}
\begin{proof}
\begin{enumerate}[(a)]
\item Let $k\in K$ and suppose that $x\in B(P_k,r_k/4)$, i.e., $d(x,P_k)< r_k/4$.
Definition \bref{def:zone} implies that  $(\Dom^2 P)_k=\dom\bigl(P_k,\bigcup_{j\neq k}\dom\bigl(P_j,\bigcup_{i\neq j}P_i\bigr)\bigr)$. Hence, in order to
prove that $x\in (\Dom^2 P)_k$ it suffices to prove that $r_k/4\leq d(x,y)$ for all $y\in \bigcup_{j\neq k}\dom(P_j,\bigcup_{i\neq j}P_i)$. Given some $y$ in this union, it belongs to $\dom\bigl(P_j,\bigcup_{i\neq j}P_i\bigr)$ for some $j\in K\backslash\{k\}$. This fact combined with the condition $k\neq j$ imply that
\begin{equation*}
d(y,P_j)\leq d\biggl(y,\bigcup_{i\neq j}P_i\biggr)\leq d(y,P_k)\leq d(y,x)+d(x,P_k)<d(y,x)+(r_k/4).
\end{equation*}
This inequality, the triangle inequality, and the definition of $r_k$, all imply that
\begin{equation*}
r_k\leq d(P_j,P_k)\leq d(P_j,y)+d(y,x)+d(x,P_k)\leq 2d(x,y)+2r_k/4,
\end{equation*}
i.e., $d(x,P_k)<r_k/4\leq d(x,y)$ and the assertion follows. Finally, Lemma \bref{lem:DomPk}  implies that $I^{(2)}=\Dom^2 P\subseteq \Dom^{\gamma} P$ for any integer $\gamma\geq 1$, and hence $(\Dom^2 P)_k\subseteq (\Dom^{\gamma} P)_k$ for each $k\in K$.
\item
Let $p\in Q$. Since $x\notin B(Q,r)$,
a simple argument  shows that the intersection of the compact segment $[x,p]$  with the closed set $\partial B(Q,r)$  is attained at some point $y\in [x,p]$, e.g., at $y=\gamma(t)$ for $t=\inf\{u\in [0,d(x,p)]: \gamma(u)\in B(Q,r)\}$
(where $\gamma:[0,d(x,p)]\to X$, is the isometric function which maps $[0,d(x,p)]$ onto $[x,p]$). It must be that
$s:=r-d(y,p)\leq 0$. Indeed, if this inequality does not hold, then $s>0$. Hence the open ball $B(y,s)$ is well defined and for each $z$ in this ball we have $d(z,Q)\leq d(z,p)\leq d(z,y)+d(y,p)<s+d(y,p)=r$. Thus $z\in B(Q,r)$ for all $z\in B(y,s)$ and hence $y$ is in the interior of $B(Q,r)$, a contradiction to the fact that $y\in \partial(B(Q,r))$.
Therefore, since $y\in [x,p]$, we obtain $d(x,p)=d(x,y)+d(y,p)\geq d(x,y)+r$. Thus, for all $p\in Q$
\begin{equation*}
d(x,B(Q,r))=d(x,\overline{B(Q,r)})\leq d(x,y)\leq d(x,p)-r.
\end{equation*}
As a result, $d(x,B(Q,r))+r\leq \inf\{d(x,p): p\in Q\}=d(x,Q)$. To see that actually $d(x,B(Q,r))+r=d(x,Q)$, let $\epsilon>0$ be arbitrary. By the definition  of $d(x,B(Q,r))$ there exists $z\in B(Q,r)$ satisfying $d(x,z)<d(x,B(Q,r))+\epsilon$. This inequality, the triangle inequality, and the fact that $z\in B(Q,r)$ imply that $d(x,Q)\leq d(x,z)+d(z,Q)<d(x,B(Q,r))+\epsilon+r$. Since $\epsilon$ can be an arbitrary small positive number, we conclude that $d(x,Q)\leq d(x,B(Q,r))+r$. Hence $d(x,B(Q,r))+r=d(x,Q)$.
\item Suppose that we know that the components of $\Dom^{3}P $ are disjoint.
Then for each integer $\gamma\geq 2$ the components of $\Dom^{\gamma}P $ are disjoint,  because $\Dom^{\gamma}P\subseteq \Dom^{3}P$ by Lemma  \bref{lem:DomPk}. We now prove that indeed the components of $\Dom^{3}P $ are disjoint.
Let $k_1\neq k_2$ be two indices in $K$, and assume to the contrary that $x\in (\Dom^3 P)_{k_1}\bigcap (\Dom^3 P)_{k_2}$.
By the definition of $\Dom$ and by Part \beqref{item:B(P,r)Dom2} we have
\begin{align*}
d(x,P_{k_1})\leq d\biggl(x,\bigcup_{j\neq k_1}(\Dom^2 P)_j\biggr) \leq d(x,(\Dom^2 P)_{k_2})\\
\leq d(x,B(P_{k_2},r_{k_2}/4))\leq d(x,P_{k_2}).
\end{align*}
Hence $x\notin \overline{P_{k_2}}$, because if $x\in \overline{P_{k_2}}$, then from the above inequality $d(x,P_{k_1})\leq d(x,P_{k_2})=0$, namely $x\in \overline{P_{k_1}}$, a contradiction to the assumption that $0<r_{k_1}\leq d(P_{k_1},P_{k_2})$. Thus $d(x,B(P_{k_2},r_{k_2}/4))< d(x,P_{k_2})$ by assumption, so actually $d(x,P_{k_1})<d(x,P_{k_2})$. In the same way $d(x,P_{k_2})< d(x,P_{k_1})$, a contradiction.

Finally, suppose that $X$ is a geodesic metric space.
Let $j,k\in K$ be different. Since $(\Dom^{\gamma}P)_i\subseteq (\Dom^3 P)_i$ for each integer $\gamma\geq 2$ and each $i\in K$ (by Lemma  \bref{lem:DomPk}), it suffices to show that $(r_k/8)+(r_j/8)\leq d(x,y)$ for all $x\in (\Dom^3 P)_{k}$ and $y\in(\Dom^3 P)_j$.
 By definition, the triangle inequality and parts \beqref{item:B(P,r)Dom2}, \beqref{item:r},
\begin{align*}
d(x,P_k)\leq d\biggl(x,\bigcup_{i\neq k}(\Dom^2 P)_i\biggr)& \leq d(x,(\Dom^2 P)_j)
\leq d(x,B(P_j,r_{j}/4))\\
&\leq d(x,P_j)-r_j/4\leq d(x,y)+d(y,P_j)-(r_j/4).
\end{align*}
In the same way $d(y,P_j)\leq d(y,x)+d(x,P_k)-(r_k/4)$. Now we add these two inequalities and do additional elementary arithmetic operations which imply that indeed  $(r_k/8)+(r_j/8)\leq d(x,y)$.
\item If $x\in (\Dom^{\gamma}(P))_j$, then $x\notin B(P_k,r_k/4)\subseteq (\Dom^{\gamma}(P))_k$ from Parts  \beqref{item:B(P,r)Dom2} and \beqref{item:Dom3Disjoint}. Hence Part \beqref{item:r} implies that $r_k/4\leq d(x,P_k)-d(x,B(P_k,r_k/4))\leq d(x,P_k)$.
\end{enumerate}
\end{proof}

\begin{lem}\label{lem:AntiIntersec_dom}
Let $(X,d)$ be a metric space, let $P\subseteq X$ be nonempty and suppose that $\{C_{\gamma}\}_{\gamma=1}^{\infty}$
is a family of subsets of $X$ such that $\bigcap_{\gamma=1}^{\infty} C_{\gamma}\neq
\emptyset$. If
\begin{equation}\label{eq:domInLess}
\dom\Biggl(P,\bigcap_{\gamma=1}^{\infty} C_{\gamma}\Biggr)= \overline{\Biggl\{x\in X:
d(x,P)<d\Biggl(x,\bigcap_{\gamma=1}^{\infty} C_{\gamma}\Biggr)\Biggr\}},
\end{equation}
and
\begin{equation}\label{eq:ALimAk}
d\Biggl(y,\bigcap_{\gamma=1}^{\infty} C_{\gamma}\Biggr)=\limsup_{\gamma\to \infty}d(y,C_{\gamma}),\quad \forall y\in X,
\end{equation} then
\begin{equation}\label{eq:AntiCont}
\dom\Biggl(P,\bigcap_{\gamma=1}^{\infty} C_{\gamma}\Biggr)=\overline{\bigcup_{\gamma=1}^{\infty}\dom(P,C_{\gamma})}.
\end{equation}
\end{lem}
\begin{proof}
By the antimonotonicity of $\dom(P,\cdot)$ we have $\dom\bigl(P,\bigcap_{j=1}^{\infty}
C_j\bigr)\supseteq\dom(P,C_{\gamma})$ for all $\gamma\in\N$. Since $\dom\bigl(P,\bigcap_{j=1}^{\infty}
C_j\bigr)$ is a closed set, the inclusion $\dom(P,\bigcap_{\gamma=1}^{\infty}
C_{\gamma})\supseteq\overline{\bigcup_{\gamma=1}^{\infty}\dom(P,C_{\gamma})}$ follows. For
the reverse inclusion, let $\epsilon>0$ be given and suppose that
$x\in \dom(P,\bigcap_{\gamma=1}^{\infty} C_{\gamma})$. We should prove that there
are $\gamma\in \N$ and $y\in \dom(P,C_{\gamma})$ such that $d(y,x)<\epsilon$.

By \beqref{eq:domInLess} there is $y\in X$ such that
$d(x,y)<\epsilon$ and $r:=d\bigl(y,\bigcap_{\gamma=1}^{\infty} C_{\gamma}\bigr)-d(y,P)>0$,
and by \beqref{eq:ALimAk} there is $\gamma\in\N$ large enough such that
$\bigl|d\bigl(y,\bigcap_{j=1}^{\infty} C_j\bigr)-d(y,C_{\gamma})\bigr|<r/2$. Hence
\begin{equation*}
d(y,P)+r/2<d\biggl(y,\bigcap_{\gamma=1}^{\infty} C_{\gamma}\biggr)<d(y,C_{\gamma})+r/2.
\end{equation*}
Thus $d(y,P)<d(y,C_{\gamma})$, and so $y\in \dom(P,C_{\gamma})$.
\end{proof}

\begin{prop}\label{prop:mDomM}
Let $(X,d)$ be a geodesic metric space and let $(P_k)_{k\in K}$ be a tuple of nonempty closed  sets in $X$ satisfying  $r:=\inf\{d(P_k,P_j): j,k\in K,\, j\neq k\}>0$. For each nonnegative integer $\gamma$
and each $k\in K$ let
\begin{equation}\label{eq:R_gammaA_k}
R^{\gamma}:=\Dom^{2\gamma+1}(P_k)_{k\in K}=O^{(\gamma)}, \quad
A_{\gamma,k}:=\bigcup_{j\neq k} R_j^{\gamma},\quad A_k:=\bigcap_{\gamma=1}^{\infty}A_{\gamma,k}.
\end{equation}
\begin{enumerate}[(a)]
\item\label{item:UnionIntersect} The following equality holds for all $k\in K$:
\begin{equation}\label{eq:UnionIntersect}
\bigcup_{j\neq  k}\bigcap_{\gamma=1}^{\infty}R_j^{\gamma}=\bigcap_{\gamma=1}^{\infty}\bigcup_{j\neq k} R_j^{\gamma}=A_k.
\end{equation}
\item\label{item:AkNonEmpty} $A_{\gamma,k}$ is a closed set  for each $k\in K$ and each $\gamma\geq 1$.
$A_k$ is nonempty, closed, and satisfies $d(P_k,A_k)\geq r/4$  for each $k\in K$.
\item\label{item:DomCapCup} If for each $k\in K$ both \beqref{eq:domInLess} and \beqref{eq:ALimAk}  hold with $C_{\gamma}=A_{\gamma,k}$, then
\begin{equation}\label{eq:DomCapCup}
\Dom\Biggl(\bigcap_{\gamma=0}^{\infty} \Dom^{2\gamma+1}(P_k)_{k\in K}\Biggr)=
\overline{\bigcup_{\gamma=0}^{\infty}\Dom^{2\gamma}(P_k)_{k\in K}}.
\end{equation}
\item\label{item:StrictlyConvexProper} If, in addition, $(X,d)$ is  proper and has the geodesic
inclusion property, then \beqref{eq:domInLess}--\beqref{eq:ALimAk} (with $C_{\gamma}=A_{\gamma,k}$, $k\in K$ arbitrary)
and hence \beqref{eq:DomCapCup} hold.
\end{enumerate}
\end{prop}

\begin{proof}
\begin{enumerate}[(a)]
\item Fix $k\in K$.
If $x\in \bigcup_{j\neq  k}\bigcap_{\gamma=1}^{\infty}R_j^{\gamma}$, then
$x\in \bigcap_{\gamma=1}^{\infty}R_j^{\gamma}$ for some $j\neq k$. Since $R_j^{\gamma}\subseteq \bigcup_{i\neq k} R_i^{\gamma}$ for all $\gamma\in\N$, it follows that
$\bigcap_{\gamma=1}^{\infty}R_j^{\gamma}\subseteq \bigcap_{\gamma=1}^{\infty}\bigcup_{i\neq k} R_i^{\gamma}$, so $x\in \bigcap_{\gamma=1}^{\infty}\bigcup_{i\neq k} R_i^{\gamma}$. On the other hand, let $x\in \bigcap_{\gamma=1}^{\infty}\bigcup_{j\neq k} R_j^{\gamma}$. Given  $\gamma\in \N$, Lemma \bref{lem:Dom3}\,\beqref{item:Dom3Disjoint} (note that $2\gamma+1\geq 3$ and
$R^{\gamma}=\Dom^{2\gamma+1}(P_k)_{k\in K}$) implies that  $R_i^{\gamma}\cap R_j^{\gamma}=\emptyset$ for indices $j\neq i$. Hence for each $\gamma\in\N$ there exists exactly one  index $j\neq k$ such that $x\in R_j^{\gamma}$. It must be that all these indices coincide. Indeed, if this is not true, then $x\in R_j^{\gamma}\cap R_{j'}^{\gamma'}$ for two natural numbers $\gamma'>\gamma$ and two indices $j'\neq j$. But  $R_{j'}^{\gamma'}\subseteq R_{j'}^{\gamma}$ by Lemma \bref{lem:DomPk}. Thus $x\in R_j^{\gamma}\cap R_{j'}^{\gamma}$, a contradiction to Lemma \bref{lem:Dom3}\,\beqref{item:Dom3Disjoint}. As a result
$x\in \bigcap_{\gamma=1}^{\infty}R_j^{\gamma}$ for some $j\neq k$, and we conclude that $x\in \bigcup_{j\neq  k}\bigcap_{\gamma=1}^{\infty}R_j^{\gamma}$. This inclusion, the one proved before, and the definition of $A_k$, all imply \beqref{eq:UnionIntersect}.

\item Lemma \bref{lem:DomPk} implies that $P_j\subseteq \bigcap_{\gamma=1}^{\infty}R_j^{\gamma}$ for each $j\in K$, and so  \beqref{eq:UnionIntersect} implies that $\emptyset\neq\bigcup_{j\neq k}P_j\subseteq \bigcup_{j\neq k}\bigcap_{\gamma=1}^{\infty}R_j^{\gamma}=A_k$ for all $k\in K$. For each $\gamma\in \N$ and $k\in K$ we have $r/4\leq d(P_k,A_{\gamma,k})\leq d(P_k,A_k)$, where the left inequality follows from Lemma \bref{lem:Dom3}\,\beqref{item:r_k} and the definition \beqref{eq:R_gammaA_k} of $A_{\gamma,k}$ (we use Lemma \bref{lem:Dom3}\,\beqref{item:r_k} with each $x\in A_{\gamma,k}$ and recall that $r\leq r_k$) and the right inequality follows from the inclusion $A_k\subseteq A_{\gamma,k}$. From Lemma \bref{lem:DomPk}
the intersection in \beqref{eq:R_gammaA_k} which defines $A_k$ is decreasing. By Lemma \bref{lem:dom}\,\beqref{item:closure} and Lemma \bref{lem:Dom3}\,\beqref{item:Dom3Disjoint} each $A_{\gamma,k}$ is a closed set, because it is a union of closed and disjoint sets with a positive distance (at least $r/4$)  between any two  different members in the union.
 Therefore $A_k$ is the intersection of closed sets and thus closed.

\item By Part \beqref{item:AkNonEmpty} we know that $A_k\neq \emptyset$ for any $k\in K$. Given $k\in K$, our assumption
 implies that \beqref{eq:domInLess}--\beqref{eq:ALimAk} hold with $C_{\gamma}=A_{\gamma,k}$ for
each $\gamma\in \N$. By Part \beqref{item:UnionIntersect} we know that \beqref{eq:UnionIntersect}
holds. These facts, together with  Lemma  \bref{lem:DomPk}, Lemma \bref{lem:AntiIntersec_dom},  the definition of $\Dom$,
and the fact that intersection, union, and closure on tuples are taken component-wise, imply
\begin{multline*}
\Dom\Biggl(\bigcap_{\gamma=0}^{\infty}\Dom^{2\gamma+1}(P_k)_{k\in K}\Biggr)
=\Dom\Biggl(\bigcap_{\gamma=0}^{\infty}R^{\gamma}\Biggr)
=\Dom\Biggl(\bigcap_{\gamma=1}^{\infty}R^{\gamma}\Biggr)\\
=\Biggl(\dom\Biggl(P_k,\bigcup_{j\neq k}\bigcap_{\gamma=1}^{\infty}R_j^{\gamma}\Biggr)\Biggr)_{k\in K}
=\Biggl(\dom\Biggl(P_k,\bigcap_{\gamma=1}^{\infty}\bigcup_{j\neq k}R_j^{\gamma}\Biggr)\Biggr)_{k\in K}\\
=\Biggl(\dom\Biggl(P_k,\bigcap_{\gamma=1}^{\infty}A_{\gamma,k}\Biggr)\Biggr)_{k\in K}
=\Biggl(\overline{\bigcup_{\gamma=1}^{\infty}\dom(P_k,\bigcup_{j\neq k}R_j^{\gamma})}\Biggr)_{k\in K}\\
=\overline{\Biggl(\bigcup_{\gamma=1}^{\infty}\Biggl(\Dom(R^{\gamma})\Biggr)_k\Biggr)_{k\in K}}
=\overline{\bigcup_{\gamma=1}^{\infty}\Dom(R^{\gamma})}
=\overline{\bigcup_{\gamma=0}^{\infty}\Dom^{2\gamma}(P_k)_{k\in K}}.
\end{multline*}

\item Since for each $k\in K$ the intersection in \beqref{eq:R_gammaA_k} which defines $A_k$ is decreasing (from Lemma \bref{lem:DomPk}) of closed subsets (part \beqref{item:AkNonEmpty}),
Lemma  \bref{lem:DistIntersect} implies that $d(x,A_k)=\lim_{\gamma\to\infty}d(x,A_{\gamma,k})$ for each $x\in X$
and hence \beqref{eq:ALimAk} holds with $C_{\gamma}:=A_{\gamma,k}$ for every $\gamma\in\N$.
Because $d(P_k,A_k)\geq r/4>0$ by Part \beqref{item:AkNonEmpty} and because $(X,d)$ is a proper geodesic metric space having the geodesic inclusion property, we conclude from Theorem \bref{thm:BoundaryInterior}
(equality \beqref{eq:closure}) that \beqref{eq:domInLess} holds with $P_k$ instead of $P$ and $A_k$ instead of $A$.
From part \beqref{item:DomCapCup} we conclude that \beqref{eq:DomCapCup} holds.
\end{enumerate}
\end{proof}

Now it is possible to prove the assertions formulated in
Section \bref{sec:ConvergenceResult}.

\begin{proof}[{\bf Proof of Theorem \bref{thm:Metric}}]
Lemma  \bref{lem:dom}\,\beqref{item:DomClosure} and Lemma \bref{lem:GeneralDistance}\,\beqref{item:DomUnionIntersect} imply that
\begin{equation}\label{eq:Dom(m)M}
\Dom(m)=\Dom\Biggl(\overline{\bigcup_{n=0}^{\infty}I^{(n)}}\Biggr)
=\Dom\Biggl(\bigcup_{n=0}^{\infty}I^{(n)}\Biggr)
=\bigcap_{n=0}^{\infty}\Dom(I^{(n)})=\bigcap_{n=0}^{\infty}O^{(n)}=M.
\end{equation}
\end{proof}
\begin{proof}[{\bf Proof of Theorem \bref{thm:AlgUniConvex}}]
By assumption $(X,d)$ is a proper geodesic metric space which has
the geodesic inclusion property and  the sites satisfy \beqref{eq:d(P_k,P_j)>0}. Hence we can apply
Proposition  \bref{prop:mDomM}\,\beqref{item:StrictlyConvexProper} which implies that
\begin{equation}\label{eq:Dom(M)m}
\Dom(M)=\Dom\Biggl(\bigcap_{\gamma=0}^{\infty}\Dom^{2\gamma+1}(P)\Biggr)
=\overline{\bigcup_{\gamma=0}^{\infty}\Dom^{2\gamma}(P)}=\overline{\bigcup_{n=0}^{\infty}I^{(n)}}=m.
\end{equation}
Thus $m=\Dom(M)$. This and Theorem \bref{thm:Metric} imply
that $M=\Dom(m)=\Dom^2(M)$ and $m=\Dom(M)=\Dom^2(m)$. Hence both $m$ and $M$ are double zone diagrams. From Lemma \bref{lem:DomPk}\,\beqref{item:In R On} and Lemma \bref{lem:dom}\,\beqref{item:DomClosure} it can be seen that $m$ and $M$ are,
respectively, the least and the greatest double zone diagrams. Finally, if $|K|=2$, then
 \beqref{eq:Dom(m)M}, \beqref{eq:Dom(M)m}, and the definition of $\Dom$, imply that  $(M_1,M_2)=\Dom(m)=((\dom(P_1,m_2),\dom(P_2,m_1))$ and $(m_1,m_2)=\Dom(M_1,M_2)=(\dom(P_1,M_2),\dom(P_2,M_1))$.
Hence $m_1=\dom(P_1,M_2)$ and $M_2=\dom(P_2,m_1)$. But  the definition of $\Dom$ implies that $\Dom(m_1,M_2)=(\dom(P_1,M_2),\dom(P_2,m_1))$. Therefore one has $\Dom(m_1,M_2)=(m_1,M_2)$, namely, $(m_1,M_2)$ is a zone
 diagram.  In the same way $(M_1,m_2)$ is a zone diagram.
\end{proof}

\begin{proof}[\bf Proof of Corollary \bref{cor:CompactLimit}]
Fix $k\in K$ and for each $n\in \N\cup\{0\}$ denote $A_{n}:=O^{(n)}_k$ and $B_n:=I^{(n)}_k$. Lemma \bref{lem:DomPk} implies that $(A_n)_{n=0}^{\infty}$ is decreasing and $(B_n)_{n=0}^{\infty}$ is increasing. As a result, since $X$ is compact \beqref{eq:Mm_k} follows from Lemma  \bref{lem:HausdorffUnionIntersection}.
\end{proof}

\begin{proof}[{\bf Proof of Corollary \bref{cor:ZoneDoubleZone}}]
If the least and the greatest double zone diagram coincide,
then, without any restriction on the sites or the space, there exists a unique zone diagram  which
is equal to both of them \cite[Corollary 6.2]{ReemReichZone}.
 Since in the proof of \cite[Theorem 1.1]{KMT2012}
it was proved that the  least double zone diagram coincides with the greatest one
when $(X,d)$ is a finite dimensional Euclidean space, the assertion follows from Theorem \bref{thm:AlgUniConvex}.
\end{proof}

\section{Concluding remarks and open problems}\label{sec:end}
We conclude the paper with the following remarks.
\begin{remark}
We believe that the methods and ideas presented in this paper can be extended to other implicit geometric objects. One of the reasons behind this belief is the fact that an important component in the iterative algorithm described in this paper for computing zone and double zone diagrams (Sections \bref{sec:QualitativeDescription}, \bref{sec:ActualCompute} above) is the approximation algorithm for computing Voronoi diagrams of general sites introduced in \cite{ReemISVD09},  and the computation of other implicit geometric objects, such as $k$-sectors, makes heavy use of Voronoi diagrams of general sites \cite[Section 4.1 and Proposition 5]{ImaiKawamuraMatousekReemTokuyamaCGTA}.
\end{remark}
\begin{remark}
A very interesting open problem is to establish error estimates for  the convergence rate (speed) of $(I^{(n)})_{n=0}^{\infty}$ and  $(O^{(n)})_{n=0}^{\infty}$. One of the reasons that this problem is interesting, is the observed fast convergence of these sequences (usually 4-5  iterations suffice for obtaining a very good approximation). We feel that the approach and error estimates developed in  \cite{ReemGeometricStabilityArxiv} may help here. We also feel,  but not entirely sure, that the rate of convergence is geometric (that is, linear).
\end{remark}
\begin{remark}
The availability of the implementation \cite{Vdream2017web} enables one to preform various interesting experiments and to discover new and unexplained phenomena. For instance, although in general a zone diagram is not necessarily unique, even in the case of $\R^2$ with two point sites (see e.g., \cite[Sections 1,5]{KMT2012} or Example \bref{ex:mM} above), and although our theoretical results, when restricted to normed spaces, are limited to normed spaces which are strictly convex, experiments show that Algorithm \bref{alg:InOn} actually converges to a unique zone diagram most of the times even when the considered normed space is not strictly convex.  This claim is illustrated in Figures \bref{fig:Zone2D-10Sites-1inSite-L1-UnitSphere005-It2-Endpoints}--\bref{fig:ZoneLinfty}.

Another interesting phenomenon which has been observed experimentally is the geometric stability of zone diagrams with respect to certain changes. More precisely, two types of stability have been observed. The first type is that small perturbations in the shapes of the sites (e.g., due to a translation or a distortion) lead to small change in the shapes of the regions of the corresponding zone diagram (formally, the changes are measured with respect to the Hausdorff distance). Figures \bref{fig:Zone2D-5Sites-55inSite-ShapesType3-L2-UnitSphere005-It5-Rays}--\bref{fig:Zone2D-5Sites-55inSite-ShapesType2-L2-UnitSphere005-It5-Rays} illustrate this phenomenon. Similarly, small changes in the norm of the space (where the sites are fixed) yield small changes in the shapes of the regions of the zone diagram, as illustrated in  Figures \bref{fig:Zone2D-10Sites-1inSite-L1-UnitSphere005-It2-Endpoints}--\bref{fig:Zone2D-10Sites-1inSite-L1-UnitSphere0005-It6-Rays}   comparing to Figure  \bref{fig:IntroZone} and also in Figures \bref{fig:Zone2D-5Sites-55inSite-ShapesType2-L2-UnitSphere005-It5-Rays}--\bref{fig:Zone2D-5Sites-55inSite-ShapesType2-L_alpha_delta-UnitSphere005-It5-Rays}. Here it is worth saying a few words on the formal way to measure these changes, since some caution is needed. Indeed, if we want to measure the changes with respect to the Hausdorff distance, then we face a problem because the Hausdorff distance depends on the norm and what we change is the norm itself. A possible solution to this issue is to consider all of the involved sets, namely the sites and the regions before and after the perturbations, as being sets in the Euclidean plane, and to  measure the changes with respect to the Hausdorff distance induced by this Euclidean norm. This choice is natural since in the real world, when we compare the difference between the shapes, we simply consider them, using our eyes, as sets located on the flat page on  which the pictures (before and after the perturbations) are embedded, and due to the Euclidean nature of our world, this flat page is modeled by the Euclidean plane.

Anyway, it is an open problem to explain the above-mentioned phenomena, and we believe that the approach presented in \cite{ReemGeometricStabilityArxiv,ReemVorStabilityNonUC2012} can help here.
\end{remark}

\begin{figure}

\begin{minipage}[t]{0.45\textwidth}
\begin{center}
{\includegraphics[scale=0.5]{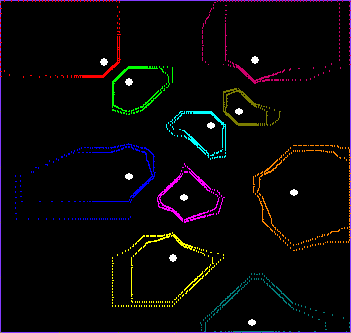}}
\end{center}
 \caption{$I^{(2)}$ and $O^{(2)}$ of the same point sites as in Figure \bref{fig:IntroZone}, now with respect to the $\ell_1$ norm. Here 160 endpoints approximate the boundary of each region. }
\label{fig:Zone2D-10Sites-1inSite-L1-UnitSphere005-It2-Endpoints}
\end{minipage}
\begin{minipage}[t]{0.45\textwidth}
\begin{center}
{\includegraphics[scale=0.5]{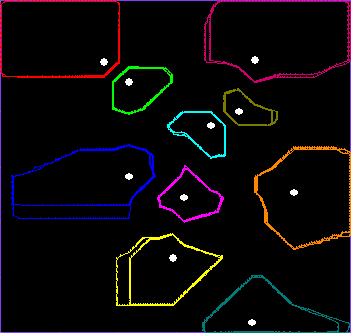}}
\end{center}
 \caption{The setting of Figure \bref{fig:Zone2D-10Sites-1inSite-L1-UnitSphere005-It2-Endpoints}, where now $I^{(3)}$ and $O^{(3)}$  are displayed and neighbor endpoints are connected by a line segment. }
\label{fig:Zone2D-10Sites-1inSite-L1-UnitSphere005-It3-EndpointsConnect}
\end{minipage}
\hfill
\begin{minipage}[t]{0.45\textwidth}
\begin{center}
{\includegraphics[scale=0.5]{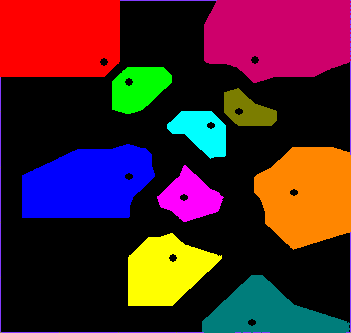}}
\end{center}
 \caption{A good approximation of the zone diagram of Figures \bref{fig:Zone2D-10Sites-1inSite-L1-UnitSphere005-It2-Endpoints}--\bref{fig:Zone2D-10Sites-1inSite-L1-UnitSphere005-It3-EndpointsConnect} (six iterations).}
\label{fig:Zone2D-10Sites-1inSite-L1-UnitSphere0005-It6-Rays}
\end{minipage}
\begin{minipage}[t]{0.45\textwidth}
\begin{center}
{\includegraphics[scale=0.5]{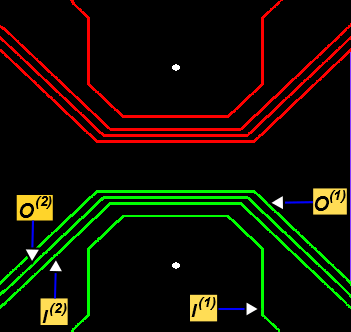}}
\end{center}
 \caption{$I^{(1)}$, $I^{(2)}$, $O^{(1)}$, and $O^{(2)}$ for $P_1=\{(0,3)\}$ and $P_2=\{(0,-3)\}$ with respect to the $\ell_{\infty}$ norm (the boundaries of the second component are labeled). In agreement with \cite[Example 2.4]{ReemReichZone}.}
\label{fig:ZoneLinfty}
\end{minipage}
\end{figure}

\begin{figure}

\begin{minipage}[t]{0.45\textwidth}
\begin{center}
{\includegraphics[scale=0.5]{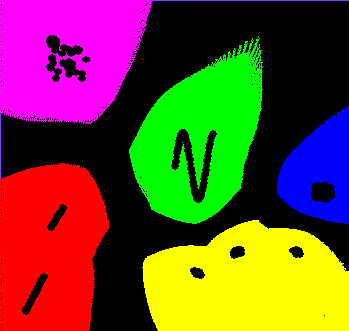}}
\end{center}
 \caption{An approximation (five iterations) of the zone diagram of five complicated sites in a square in $(\R^2,\ell_2)$. Here 160 rays approximate  each region. The approximation is very good since each site consists of many points and hence the emanating rays cover most of the regions.}
\label{fig:Zone2D-5Sites-55inSite-ShapesType3-L2-UnitSphere005-It5-Rays}
\end{minipage}
\begin{minipage}[t]{0.45\textwidth}
\begin{center}
{\includegraphics[scale=0.5]{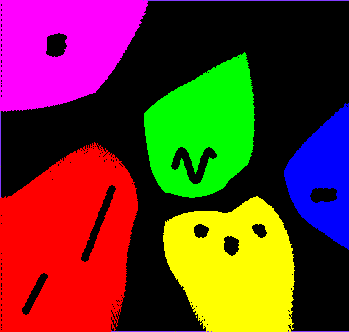}}
\end{center}
 \caption{The setting of Figure \bref{fig:Zone2D-5Sites-55inSite-ShapesType3-L2-UnitSphere005-It5-Rays}, with perturbed sites. The change in the zone diagram is not very big, but it is not negligible (e.g., near the bottom right corner, since the perturbation of some sites is not negligible).}
\label{fig:Zone2D-5Sites-55inSite-ShapesType4-L2-UnitSphere005-It5-Rays}
\end{minipage}
\hfill
\begin{minipage}[t]{0.45\textwidth}
\begin{center}
{\includegraphics[scale=0.5]{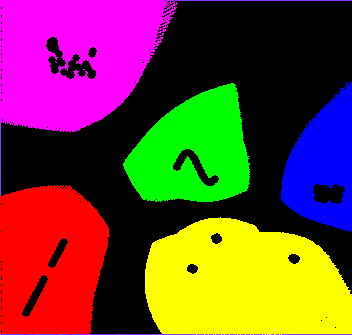}}
\end{center}
 \caption{The setting of Figure \bref{fig:Zone2D-5Sites-55inSite-ShapesType3-L2-UnitSphere005-It5-Rays}, with other perturbed sites. Despite this perturbation, the zone diagram has not varied much.}
\label{fig:Zone2D-5Sites-55inSite-ShapesType2-L2-UnitSphere005-It5-Rays}
\end{minipage}
\begin{minipage}[t]{0.45\textwidth}
\begin{center}
{\includegraphics[scale=0.5]{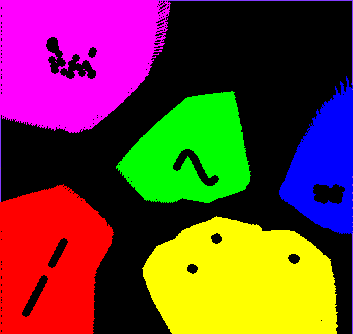}}
\end{center}
 \caption{The setting of Figure \bref{fig:Zone2D-5Sites-55inSite-ShapesType2-L2-UnitSphere005-It5-Rays}, where now the norm is defined in \beqref{eq:alpha delta} with $\alpha=0.7$, $\delta=0.5$. Despite the change in the norm, the zone diagram varied only slightly.}
\label{fig:Zone2D-5Sites-55inSite-ShapesType2-L_alpha_delta-UnitSphere005-It5-Rays}
\end{minipage}
\hfill
\end{figure}

\begin{remark}
It is interesting to find out whether Theorem \bref{thm:AlgUniConvex} can be generalized to other settings, for instance to all proper geodesic metric spaces, or at least to all finite dimensional normed spaces and to some ``nice'' manifolds.
\end{remark}
\begin{remark}
It will be interesting to extend the theory of geodesic metric spaces which has the geodesic inclusion property (Section \bref{sec:Geodesic} above) and, in particular, to give more examples of geodesic metric spaces which have this property. For instance to prove or disprove whether locally compact, complete and connected smooth (or piecewise smooth) Riemannian manifolds \cite{Jost2011book}, as well as Hadamard spaces \cite{Bacak2014book} (in particular, hyperbolic spaces \cite[p.\,10]{Bacak2014book}) have this property.
\end{remark}

\begin{remark}\label{rem:WeakSolution}
In Section \bref{sec:Intro} we discussed briefly some similarities and differences between  implicit computational geometry and differential equations. In what follows we mention another similarity: the relation between double zone diagrams and  the concept of weak solutions. A weak solution to a partial differential equation \cite[pp.\,201,\,221,\,292 and elsewhere]{Brezis2011book} is a function that does not satisfy the original equation but rather it satisfies a weaker condition which any true solution (also called ``a strong'' or ``a classical'' solution) to the original equation (if such a solution exists) must satisfy. Hence a weak solution generalizes the concept  of a strong solution. It is useful to consider weak solutions since in many cases it is possible to prove their existence. Once this is done, the weak solution becomes a  candidate to be a strong solution. Now one needs to prove that the weak solution is indeed a strong solution, and in many cases this can be done.  Returning to zone and double zone diagrams, we see that a possible technique to prove the existence of zone diagrams is to consider first a certain  double zone diagram (as was already said before, the existence of double zone diagrams is known in a general setting \cite[Theorem 5.5]{ReemReichZone}) and then to show that such a double zone diagram must in fact be a zone diagram. This technique has been used in various works \cite{ReemReichZone,KMT2012}. Hence, in some sense double zone diagrams can be regarded as weak solutions to the zone diagram equation $R=\Dom(R)$.
\end{remark}

\section*{Acknowledgments}
This paper is the output of a long and challenging iterative process. Any honest and constructive feedback that has been given to me  during this process is appreciated. I want to use this opportunity to thank the editors of DCG for their patience and to thank the referee for valuable comments. I also want to thank Mario Ponce for inviting me for a visit in the Pontifical Catholic University of Chile (Santiago) during June 2015, for helpful discussions, warm hospitality, and for allowing me to speak there (and in the Pontifical Catholic University of  Valpara\'iso) on issues related to implicit computational geometry in general and to this paper in particular.
A very sad moment along the way was when I discovered that Ji{\v{r}}{\'{\i}} (Jirka)  Matou{\v{s}}ek, one of the founders of implicit computational geometry, passed away. Unfortunately for me, I have not had the opportunity to meet him, but at least I had the honor to communicate with him electronically and to be a coauthor of him. It was through these  communications and also through his scientific outcome that I have discovered his exceptional abilities. I hope and believe that his contributions will have a lasting value and I dedicate this paper to him.

Parts of this work were done during various years, when I have been associated with the following places: The Technion, Haifa, Israel (2009-2010, 2016-2017), the University of Haifa, Haifa, Israel (2010-2011), the National Institute for Pure and Applied Mathematics  (IMPA), Rio de Janeiro, Brazil (2011-2013), and the Institute of Mathematics and Computational Sciences (ICMC), University of S\~ao Paulo, S\~ao Carlos, Brazil (2014-2015), and it is an opportunity for me to thank the BSF and FAPESP. Special thanks are for a special postdoc fellowship (``P\'os-doutorado de Excel\^encia'') given to me when I was at IMPA.

\bibliographystyle{spmpsci}

\end{document}